\documentclass[pdflatex,sn-mathphys-num]{sn-jnl}

\usepackage{makecell}
\usepackage{tabularray}
\usepackage{graphicx}
\usepackage{multirow}
\usepackage{amsmath,amssymb,amsfonts}
\usepackage{amsthm}
\usepackage{mathrsfs}
\usepackage[title]{appendix}
\usepackage{xcolor}
\usepackage{ulem}
\usepackage{textcomp}
\usepackage{manyfoot}
\usepackage{booktabs}
\usepackage{algorithm}
\usepackage{algorithmicx}
\usepackage{algpseudocode}
\usepackage{listings}
\usepackage{xcolor}
\usepackage{amsmath,mathtools,amssymb,amsthm}
\usepackage{tikz}
\usepackage{physics}
\usepackage{comment}
\usepackage{booktabs}
\usepackage{tabularx}

\theoremstyle{thmstyleone}
\newtheorem{theorem}{Theorem}

\theoremstyle{thmstyletwo}

\theoremstyle{thmstylethree}
\newtheorem{definition}{Definition}

\begin{document}

\title{Scalable quantum circuit generation for iterative ground state approximation using Majorana Propagation}

\author[1]{\fnm{Rahul} \sur{Chakraborty}}\equalcont{These authors contributed equally to this work.}
\author[1]{\fnm{Aaron} \sur{Miller}}\equalcont{These authors contributed equally to this work.}
\author[1]{\fnm{Anton} \sur{Nyk\"anen}}\equalcont{These authors contributed equally to this work.}
\author*[1]{\fnm{\"Ozlem} \sur{Salehi}}
\equalcont{These authors contributed equally to this work.} \email{ozlem.salehi@algorithmiq.fi}
\author[1]{\fnm{Fabio} \sur{Tarocco}}
\author[1]{\fnm{Fabijan} \sur{Pavo\v{s}evi\'c}}
\author[1]{\fnm{Pi A. B.} \sur{Haase}}
\author[1]{\fnm{Martina} \sur{Stella}}
\author[1]{\fnm{Adam} \sur{Glos}}

\affil[1]{\textit{Algorithmiq Ltd, Kanavakatu 3 C, FI-00160 Helsinki, Finland}}

\abstract{ 
We introduce the Adaptive Derivative-Assembled Pseudo-Trotter ansatz Variational Majorana Propagation Eigensolver (ADAPT-VMPE), a quantum-inspired classical algorithm that exploits Majorana Propagation (MP) to produce circuits for approximating the ground state of molecular Hamiltonians. Equipped with the theoretical guarantees of MP, which provide controllable bounds on the approximation error, ADAPT-VMPE offers an efficient and scalable approach for iterative ansatz construction. A theoretical analysis of the computational complexity demonstrates that it is polynomial in both the number of qubits and the number of iterations. We present an in-depth analysis of circuit construction strategies, analyzing their impact on convergence and provide practical guidance for efficient ansatz generation. Using ADAPT-VMPE, we construct up to 100-qubit ans\"atze for a strongly correlated photosensitizer currently undergoing human clinical trials for cancer treatment. Our results demonstrate that constant overlap with the ground state across system sizes can be reached in polynomial time with polynomially sized circuits.}

\maketitle

Understanding and accurately modeling quantum many-body systems lies at the heart of many fundamental problems in physics and chemistry. This task is notoriously difficult for classical computers due to the exponential growth of the underlying Hilbert space with increasing system size. Quantum computing offers a promising alternative for overcoming this limitation by leveraging quantum mechanical principles to efficiently represent and simulate complex quantum systems that are otherwise intractable on classical hardware~\cite{cao2019quantum, mcardle2020quantum}. Among the various application domains of quantum computing, quantum chemistry is often regarded as one of the most promising candidates for achieving practical quantum advantage \cite{lee2023evaluating, hoefler2023disentangling,chan2024quantum}, as those problems naturally fit quantum hardware~\cite{aspuru2005simulated,elfving2020will}. 

In quantum chemistry, most applications ultimately reduce to solving the electronic structure problem \cite{helg00}, which involves estimating ground- and excited-state properties of a Hamiltonian describing the system's dynamics. While numerous classical approaches exist, qubit-based quantum computing offers alternative strategies. A common approach is to prepare a quantum circuit that approximates the ground state of the Hamiltonian. One of the earliest and most prominent algorithms for this task is the Quantum Phase Estimation (QPE) algorithm~\cite{kitaev1995quantum}. Given an initial state with non-zero overlap with an eigenstate of the Hamiltonian, QPE can estimate the corresponding eigenvalue with high precision while projecting the state onto that eigenstate. In particular, if the initial state has sufficient overlap with the ground state, QPE can yield an accurate estimate of the ground-state energy. However, the deep circuits required by QPE make it impractical for near-term quantum hardware. As an alternative, the Variational Quantum Eigensolver (VQE)~\cite{peruzzo2014variational} and its extensions~\cite{aspuru2005simulated, lee2018generalized, motta2023bridging, kandala2017hardware, burton2024accurate, grimsley2019adaptive, ryabinkin2018qubit, ryabinkin2020iterative} were proposed for the noisy intermediate-scale quantum (NISQ) era~\cite{preskill2018quantum, zimboras2025myths}. VQE is a hybrid quantum--classical algorithm in which a parameterized quantum circuit (ansatz) prepares trial states, while a classical optimizer iteratively updates the parameters to minimize the expectation value of the Hamiltonian measured on a quantum device.

On a fault-tolerant quantum computer, QPE could potentially provide an exponential advantage over classical algorithms for ground state preparation. However, this framework presents a challenge which is often overlooked. The success probability of the algorithm depends on the overlap between the initial state and the target eigenstate. While the Hartree-Fock determinant overlaps well enough with the ground state for small systems, this may not be the case for large or correlated systems. The exponential decay of the overlap is known as the \textit{orthogonality catastrophe}~\cite{tubman2018postponing} and has significant consequences. For instance, in Ref.~\cite{lee2023evaluating}, the authors show that two different initial state preparation methods based on Slater determinants can yield only $10^{-7}$ overlap for FeMoCo, an iron sulfur cluster containing 8 transition metal atoms. This would require $10^7$ repetitions of the QPE procedure, which would take a lot of time on quantum hardware. This presents the need for initial state preparation methods that are able to produce at least polynomially decreasing overlap in the system size. On the other hand, such initial state preparation methods should be efficient, that is, they should run in polynomial time, to enable significant quantum advantage. 

Heuristic state preparation methods are therefore likely to become essential to quantum simulation. Originally conceived in the context of near-term quantum computing, where resources are particularly scarce, VQE can prove a valuable source of inspiration for such heuristics. However, despite its promise for near-term quantum hardware, VQE faces several fundamental challenges that limit its practical deployment on current quantum devices. Hardware constraints such as noise, limited qubit connectivity, and statistical errors associated with the estimation of expectation values can significantly degrade performance and hinder convergence, particularly for large-scale systems~\cite{tilly2022variational}. Moreover, as the system size or circuit depth increases, the gradient of the cost function can vanish exponentially, a phenomenon known as barren plateaus~\cite{wang2021noise, larocca2025barren}, making the classical optimization step extremely difficult. In addition, each step requires the measurement of complex observables to high precision, a notoriously difficult task on quantum computers~\cite{garcia2021learning,cao2019quantum}, making VQE costly not only in the NISQ regime but also in the fault-tolerant (FT) era. Under these constraints, efficient and noise-robust ground state preparation becomes critical in the NISQ regime, as limited coherence times and high error rates restrict circuit depth, and poor states can further amplify measurement overhead and optimization difficulties in variational algorithms.

Consequently, developing scalable ground state preparation methods using only classical resources is a foundational requirement for achieving practical quantum advantage across both near-term and fault-tolerant quantum computers. In the literature, various methods of this kind have been proposed. Ground state preparation methods that synthesize quantum circuits based on the sum of Slater determinants~\cite{tubman2018postponing, fomichev2024initial} or matrix product state (MPS)~\cite{khan2023pre, malz2024preparation} have been mostly considered for molecules of small size, or require a vast amount of resources~\cite{berry2025rapid}. In Ref.~\cite{mullinax2025large, mullinax2025classical, hirsbrunner2024beyond}, the authors consider sparse wave-function simulation (SWCS) for directly constructing circuits. In addition, the unitary cluster Jastrow ansatz~\cite{motta2023bridging} can be warm-started based on Coupled Cluster Singles and Doubles~\cite{Bartlett1982} calculations that are routinely feasible on classical computers for small and moderately sized molecular systems. Classical simulation of the iterative qubit coupled cluster (iQCC) method and variational double bracket flow (vDBF)~\cite{steiger2024sparse, ryabinkin2025optimization, genin2025towards, shrikhande2025rapid} classically evolves the Hamiltonian to iteratively construct an ansatz. Recently, new experiments with iQCC~\cite{genin2025towards, jenab2026parallel} and with vDBF~\cite{shrikhande2025rapid} were made for up to $\approx 200$ and 128 qubits, respectively.

In this paper, we introduce a novel state preparation algorithm, the ADAPT Variational Majorana Propagation Eigensolver (ADAPT-VMPE). ADAPT-VMPE is an efficient and customizable classical algorithm for ground state preparation. It builds an ansatz iteratively, in the spirit of the ADAPT-VQE algorithm~\cite{grimsley2019adaptive}, by growing the ansatz one gate at a time, selecting elements that result in the greatest improvement from a predefined pool of gates. ADAPT-VMPE runs entirely in Fermionic space and thus constructs a Fermionic ansatz, delaying the choice of the Fermion-to-qubit (F2Q) mapping until the circuit is executed on a quantum computer. At the gate selection and parameter optimization stages, ADAPT-VMPE removes the need for quantum resources by exploiting the Majorana Propagation (MP) algorithm~\cite{miller2025simulation} for approximating the expectation value of the given Hamiltonian with respect to the quantum circuit. 
The approximation error arising from MP is theoretically guaranteed to decay exponentially in its cutoff parameter for unstructured circuits, and computing expectation values to a fixed accuracy requires only polynomial resources. \textit{This makes ADAPT-VMPE a scalable algorithm whose convergence stands on solid theoretical grounds.} This favorable scaling enables studying more complex molecular systems and supports hybrid classical-quantum strategies, suggesting potential for quantum advantage. 

In the NISQ era, ADAPT-VMPE can be viewed as a stand-alone ground state energy approximation algorithm when combined with some post-processing methods such as the Virtual Linear Map Algorithm (VILMA) to improve convergence to the target state~\cite{2022VILMA}. Yet the applicability of ADAPT-VMPE is not limited to the NISQ era, as it can also be used as a ground state preparation algorithm for algorithms like QPE in the FT era. ADAPT-VMPE can generate high-overlap states at attainable circuit cost, which makes it a strong candidate for ground state preparation. We showcase our results for a derivative of the TLD-1433 molecule~\cite{McFarland2019}, a ruthenium-based (Ru(II)) bipyridine complex currently undergoing human clinical trials for non-muscle invasive bladder cancer, which we refer to as Complex A throughout the text. In particular, we show that constant overlap with the ground state across system sizes can be reached in polynomial time, demonstrating the scalability of our method.

Our scheme distinguishes itself from existing iterative ground-state preparation methods in several ways. It is tailored to construct an ansatz in the Fermionic space, which has proven to have better convergence properties than those constructed in the qubit space~\cite{miller2024treespilation}. The use of Majorana Propagation allows efficient approximation of multiple expectation values, enabling optimization of all parameters in the ansatz at each iteration. This in turn allows the incorporation of orbital rotations~\cite{burton2022exact,fitzpatrick2024self}, in the present case restricted to the active space, as a part of the ansatz, significantly enhancing the performance of the algorithm, which from now on we will call active rotations. ADAPT-VMPE also allows the placement of newly selected gates at any point in the ansatz, providing essential performance improvements for certain molecules. Regarding these aspects, we present a comparison highlighting the advantages of our method against other iterative ansatz construction algorithms.

\section{Results}\label{sec:results}
\begin{figure}
    \centering
    \includegraphics[width=\linewidth]{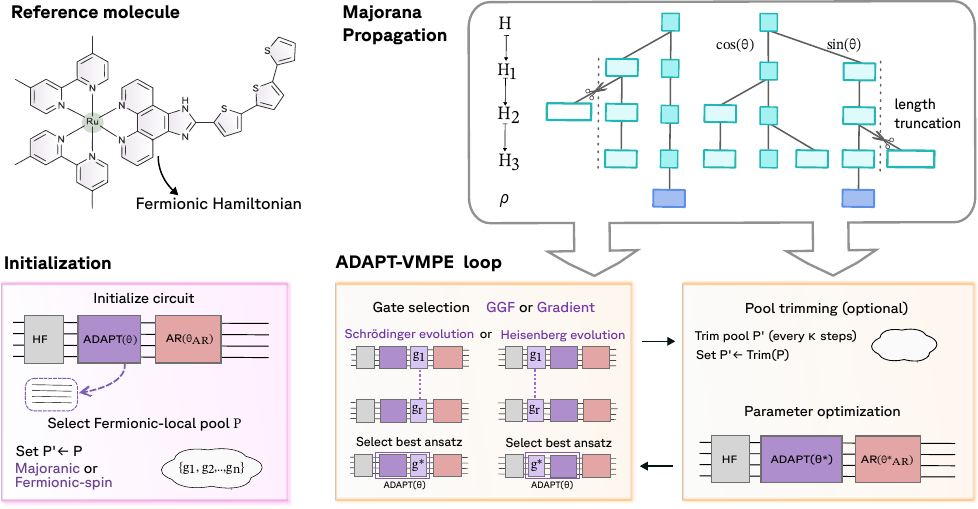}
    \caption{\textbf{Overview of the ADAPT-VMPE algorithm for constructing a fermionic ansatz from a molecular Hamiltonian.} The algorithm is initialized with Hartree–Fock state, followed by active rotations. A Fermionic-local operator pool is defined, and the method proceeds iteratively. At each iteration, operators in the pool are evaluated using GGF or gradient selection, and the selected operator is appended to the circuit in either the Schrödinger or Heisenberg representation. Pool trimming may be applied to reduce the number of candidate operators. After each addition, all variational parameters are optimized before the next iteration. The whole pipeline works in polynomial time thanks to using Majorana Propagation in gate selection and parameter optimization.}
    \label{fig:placeholder}
\end{figure}
\subsection{Introduction to ADAPT-VMPE} 

Given a second-quantized Hamiltonian, state preparation algorithms aim to construct a quantum state that approximates the ground state. Among the various classical, quantum, and hybrid methods for state preparation, ADAPT-VMPE falls into the category of those that construct a quantum state iteratively. In this category, one typically starts with a simple ansatz (often the one implementing the Hartree-Fock determinant), and gradually adds parameterizable gates in order to allow convergence to the ground state.
Several key components and properties that impact both the convergence and performance can be identified. First and foremost, the pool must be specified, a set of gates that can be added to the ansatz. Even once such a pool is specified, one still needs to decide by what criterion the best gate or gates are selected, and where the gate should be placed in the current ansatz. Once the ansatz is temporarily fixed, some of the parameters may be optimized to further decrease the energy of the corresponding state. In the following paragraphs, we focus on some aspects of the iterative procedure.

\subparagraph{Pool selection}  The choice of pool plays a vital role: we can roughly split the most commonly used pools into two categories. The first comprises gates that directly generate a physical evolution of the wavefunction~\cite{grimsley2019adaptive,tsuchimochi2022adaptive,miller2024treespilation}. Such gates are typically defined compactly in the Fermionic space via excitation and annihilation operators or via Majorana operators, which is why we call them Fermionic-local. The second category consists of gates that are qubit-local operations: those gates act only on a bounded number of qubits (typically no more than four), guaranteeing efficient implementation~\cite{tang2021qubit,yordanov2021qubit,ramoa2025reducing}. Usually, they are constructed by taking Fermionic-local operators in the Jordan-Wigner mapping and dropping the long $Z$-chain operators. While looking simple in qubit space, they no longer describe a simple physical evolution, eventually decreasing the performance: this can be observed during ADAPT-VQE optimization as presented in Fig.~\ref{fig:Fermionic-vs-nonFermionic}(a), where the Fermionic-spin pool outperforms its qubit-local counterpart QEB, and the Majoranic pool outperforms its qubit-local counterpart qubit-ADAPT. 

Qubit-local pools were introduced mostly because their implementation on a quantum computer is very efficient, requiring only $\order{1}$ CNOTs per gate. However, using recent findings, this advantage is largely eliminated: for all-to-all connectivity, selecting the optimal ternary-tree mapping~\cite{jiang2020optimal} imposes a requirement of only $\order{\log N}$ per Fermionic-local gate. The difference is even less striking for limited connectivity: for 2D structures like grids or heavy-hexagonal graphs, nodes are typically at average node-distance $\order{\sqrt N}$ for $N$-node graphs, requiring costly SWAPs of similar order. The same complexity can be provably achieved with an appropriate F2Q mapping~\cite{miller2023bonsai}. Taking this into account alongside the convergence rate, one can conclude that Fermionic-local operators are the preferable choice.

\subparagraph{Gate selection and parameter optimization} The type of gates is not the only thing to specify; it is also essential to decide how gate is selected and where they should be placed at each step. In the original ADAPT-VQE paper, as well as in many subsequent variants, the quality of the gates is evaluated by computing the gradient at parameter $\theta=0$, as if the gate would be placed at the end of the ansatz (which we call ansatz evolution in the Schr\"odinger picture)~\cite{grimsley2019adaptive}. This choice is important for variants of ADAPT-VQE that rely on a quantum computer: it can be implemented by measuring the quantum state on a quantum computer and estimating the expectation value of the commutator of the gate generator and the Hamiltonian. However, when classical simulators are used, we can be much more flexible in this choice. For example, in Ref.~\cite{brown2025iterative,shrikhande2025rapid}, gates were always added at the beginning of the ansatz in order to fit the simulator used, which performed the evolution in the Heisenberg picture. The impact of this choice can already be seen for the 16-qubit Hydrogen chain as presented in Fig.~\ref{fig:Fermionic-vs-nonFermionic}(b), and has also been recently analyzed in depth~\cite{stadelmann2025strategies}, which, following the intuition, shows that extending the ansatz in the Schr\"odinger picture performs better. It is worth noting that this may depend on the molecule, and e.g. in Fig.~\ref{fig:Fermionic-vs-nonFermionic}(d) same comparison done for Complex A is not showing any difference between Schr\"odinger and Heisenberg evolution. Furthermore, instead of selecting the gate based on the gradient, one can very efficiently estimate the maximal energy improvement of each gate, e.g., via greedy gradient-free (GGF) selection~\cite{feniou2025greedy} for most of the pools. A comparison between gradient selection and GGF is presented in Fig.~\ref{fig:Fermionic-vs-nonFermionic}(d), where the latter outperforms the former.

\begin{figure}[h]
\includegraphics[width=\linewidth]{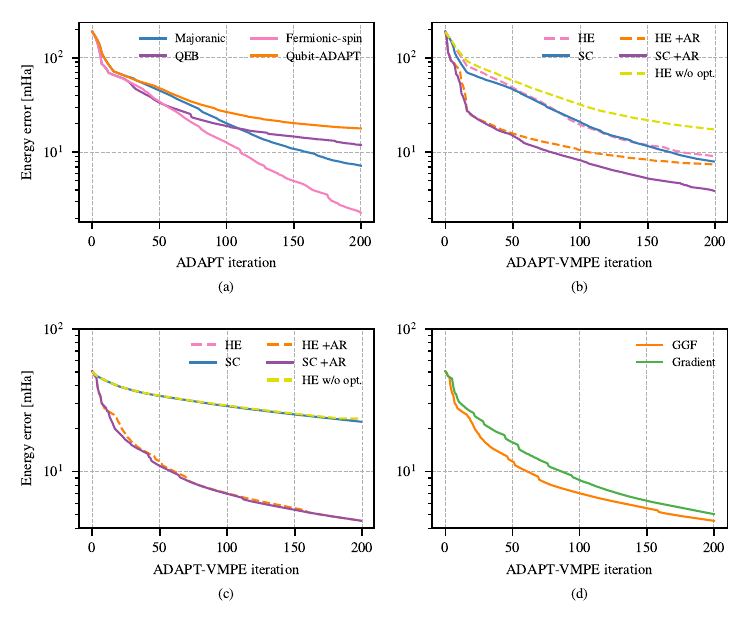}
\caption{\textbf{ADAPT-VQE simulations comparing the different strategic choices for the iterative ansatz generation procedure.} (a) ADAPT-VQE simulations performed with different pools for 16 qubit H-chain molecule. (b) and (c) ADAPT-VQE simulations run with Majorana Propagation simulator performed in different settings for 16 qubit H-chain (b) and 40 qubit Complex A (c). Heisenberg and Schrödinger versions have been considered with and without active rotations. The whole ansatz is optimized at each step except the `HE w/o AR w/o opt.' where no optimization is performed at all. In all simulations, the GGF selection method is used for gate selection, and the L-BFGS-B algorithm is used for optimization. (d) ADAPT-VQE simulation run with Majorana Propagation simulator for 40 qubit Complex A. The Heisenberg version has been considered with active rotations. The whole ansatz is optimized at each iteration using the L-BFGS-B algorithm.}
\label{fig:Fermionic-vs-nonFermionic}
\end{figure}

Parameter optimization is a step typically performed for all variational algorithms; however, for adaptive ones, it has some notable differences. First, parameter optimization is typically warm-started with the values from previous iterations, greatly speeding up and simplifying the process. For newly added gates, the parameter is set either to 0 or to the optimal value minimizing the energy from the GGF evaluation. Since the latter already yields a remarkable energy improvement, it has been suggested many times in the literature that reoptimizing previously added parameters may not be needed, which would result in a substantial speed-up. While the speed-up is indeed easily observed, this strategy greatly degrades the convergence, as previously added gates still have significant potential to improve the state, as observed in Ref.~\cite{miller2024treespilation} and presented in Fig.~\ref{fig:Fermionic-vs-nonFermionic}(b). Therefore, optimizing all parameters is a highly impactful decision and should be preferred when targeting shallow circuits representing particularly low-energy states. 

Recently, active rotations were proposed as a subroutine for classical quantum chemistry methods like DMRG~\cite{FOO_DMRG}, as well as in a VQE-like context~\cite{wahtor2022}, to improve algorithm performance. While the true ground state is a fixed point for any active rotation transformation (Equations 12.2.55 and 12.2.56 in Ref.~\cite{helg00}), optimizing those rotations has a significant impact on the process of convergence to that state, as shown in Fig.~\ref{fig:Fermionic-vs-nonFermionic}(b) and (c). Being able to reoptimize all parameters at each step makes active rotations a natural part of the ansatz in ADAPT-VMPE routine. Furthermore, active rotations can be dressed to the Hamiltonian resulting in a different second quantized Hamimltonian. Hence, they do not introduce any additional gate cost during transpilation, while improving the convergence significantly.

\subparagraph{ADAPT-VMPE} Having identified many of the critical aspects of how such iterative procedures should be defined, equipped with the analysis above, we are now ready to propose ADAPT-VMPE, an iterative method using only classical resources. As the name suggests, our circuit optimization relies on Majorana Propagation, an algorithm for efficient energy estimation of the circuit. The algorithm starts with a circuit composed of the Hartree-Fock determinant circuit and a sequence of single-excitation operators, which play the role of active rotations. ADAPT-VMPE iteratively loops over a pool composed of exponentiations of generators being linear combinations of low-length Majorana monomials, such as Majoranic or Fermionic-spin pools, and selects the gate based on gradient or GGF evaluation anywhere between the Hartree-Fock part and the active rotations. This process is repeated for $K$ iterations; after each gate addition, parameter optimization is performed on the currently constructed circuit of the form
\begin{equation} \label{eq:adapt-vmpe-output}
    \ket{\psi} = \prod_{q,q'} \exp(\theta_{q,q'} (a_q^\dagger a_{q'}- a_{q'}^\dagger a_q)) \prod_{k=1}^K \exp(\mathrm  \theta_k G_k) \ket{\rm HF},
\end{equation}
where $\theta$ are optimizable parameters and $G_k$ are generators selected during the pool selection procedure. The active rotations (left product) can be extracted at any point and efficiently dressed onto the input Hamiltonian; however, since Majorana Propagation can implement single excitations without any truncation (and thus without introducing error), it can compute the expectation value of the full state $\ket{\psi}$ with respect to the Hamiltonian without requiring dressing as a separate process.

Using Majorana Propagation not only allows us to implement all the techniques mentioned above, but also provides several vital benefits. Any classical simulator necessarily induces some approximation error (a difference between the true and estimated energy of the state) if it is to operate with polynomial resources. Majorana Propagation achieves exponential approximation error suppression using only polynomial resources. In addition, by construction, the landscape generated by Majorana Propagation is smooth in the optimized $\theta$ parameters, which simplifies the optimization process. Majorana Propagation also allows constructing a surrogate graph that depends only on the ansatz structure, not on the values of $\theta$: such a graph, once constructed, can be reused for energy and gradient computation. As we show in Appendix~\ref{sec:majoranaprop}, while a surrogate graph may take 150 seconds to construct for 76 qubits and 300 gates, it requires only 1.3 seconds for computing energy and gradient. This type of performance is not available for coefficient-based truncation schemes as proposed in~\cite{mullinax2025classical, brown2025iterative}. The surrogate graph can also be easily extended in either the Heisenberg or Schr\"odinger picture: the former was explained in the original paper~\cite{miller2025simulation}, and the extension to the Schr\"odinger picture is described in Appendix~\ref{sec:majoranaprop}.

\subparagraph{Comparison against other methods} Let us compare ADAPT-VMPE against other iterative algorithms that use only classical resources: iterative qubit coupled cluster (iQCC)~\cite{ryabinkin2020iterative}, variational double bracket flow (vDBF)~\cite{shrikhande2025rapid}, and sparse wavefunction circuit solver ADAPT-VQE (SWCS)~\cite{mullinax2025classical}. iQCC is an algorithm that adaptively adds qubit-local operations, evolving the ansatz in the Heisenberg picture. The Hamiltonian is truncated by disregarding Pauli operators with small coefficients and the landscape is optimized with a surrogate function built on top of it. The pool is not fixed; instead, it is determined at each iteration based on the evolved Hamiltonian. vDBF is an algorithm similar to iQCC, except the gates considered are Fermionic-local, and only the newly added gate is optimized. Both algorithms were recently tested for large systems, at 200 and 128 qubits, respectively. SWCS presents a different concept, where it is the state that is evolved, and computational basis states are removed if the corresponding amplitudes are of small magnitude. This makes SWCS a special type of SCI algorithm aimed at constructing a circuit directly instead of just a combination of Slater determinants. The original paper employed a Fermionic-spin pool, which is Fermionic-local~\cite{mullinax2025classical}; more recently, the same algorithm has been applied using the coupled exchange operators (CEO) pool~\cite{mullinax2025classical,ramoa2025reducing}, which is qubit-local.

We summarize the differences between these methods in Table~\ref{tab:summary}, comparing all aspects discussed in the previous section. Among other differences, ADAPT-VMPE is the only method that uses active rotations to improve the convergence rate. It is worth noting that active rotations appear to be particularly difficult to combine with iQCC or vDBF, where the evolution of the ansatz is made in the Heisenberg picture; for these algorithms, elements added first to the ansatz (closest to the Hamiltonian) are no longer optimized. Because of this, one cannot place active rotations next to the Hamiltonian and reoptimize them alongside the iterative ADAPT ansatz construction. As shown in Fig.~\ref{fig:Fermionic-vs-nonFermionic}, the ability to use and optimize active rotations has an outstanding effect on the convergence rate. 

We would like to point out that neither iQCC nor vDBF works with a fixed pool; instead, both generate the pool on the fly based on the evolved Hamiltonian, a concept we call `pool synthesis' in this paper. From our analysis of the iQCC results for N$_2$~\cite{steiger2024sparse}, we can see that a large majority of the gates selected are at most 4-local Pauli exponentiations: only 78 out of 1380 gates added were defined over more than 4 qubits. These gates, with proper corrections with Z Pauli operators, can be made a 4-length Majorana monomial, e.g., a Fermionic local gate taken from a double excitation. One should mention that the presence or absence of $Z$ strings has \textit{no effect on the pool selection} when done in the Heisenberg picture, since $Z$ strings commute with the Hartree-Fock state, as already noted in~\cite{ryabinkin2020iterative}. Since pool selection in the Heisenberg picture cannot distinguish between qubit-local and Fermionic-local operations, we claim that pool synthesis would effectively select mostly Fermionic-local operations anyway, ignoring higher-local Pauli gates. The other gates selected were at most 8-local for a 56-qubit system, which would correspond to quadruple excitation Fermionic gates. Based on the conclusion that pool synthesis would select Fermionic-local operators, we claim that it fits well with ADAPT-VMPE, where Fermionic-local gates are selected for high-quality energy approximation with Majorana Propagation. Our initial investigation confirms that this technique greatly reduces the time needed for evaluating the full pool, as it can extract only the relevant gates to be added by analyzing the evolved Hamiltonian.

\newcommand{\yes}{{\Large \textcolor{black!100!green}{${\checkmark}$}}}
\newcommand{\no}{{\Large\textcolor{black!100!red}{${\times}$}}}
\begin{table}[]
    \centering
    \begin{tabular}{@{}p{2.5cm}@{}lcccc@{}}
    \toprule\small
    && iQCC & vDBF & SWCS & \makecell{ADAPT-VMPE} \\\midrule
    \multicolumn{2}{@{}c}{\makecell[l]{continuous parameter \\landscape}} & \yes  & \no & \no & \yes  \\[5pt]
    \multicolumn{2}{@{}c}{\makecell[l]{active rotations}} & \no  & \no & \no & \yes  \\[5pt]
    \multicolumn{2}{@{}c}{\makecell[l]{Fermionic-local operations}} & \no  & \yes & \yes & \yes  \\[5pt]
    \makecell[l]{evolution} && Heisenberg& Heisenberg & Schr\"odinger & \makecell{Heisenberg \&\\ Schr\"odinger}
    \\[8pt]
    \multicolumn{2}{@{}c}{\makecell[l]{pool gate selection}} & gradient  & gradient & gradient & \makecell[c]{gradient \&\\ GGF} \\[8pt]
    \multicolumn{2}{@{}c}{\makecell[l]{pool synthesis}} & \yes  & \yes & \no & \makecell{Under\\investigation} \\[8pt]
    \multirow{2.7}*{\makecell[l]{CNOT per gate\\(typical, worst)}}  & heavy-hex & $\order{\sqrt N}$, $\order{N}$ & $\order{\sqrt N}$, $\order{N}$ & $\order{\sqrt N}$ & $\order{\sqrt N}$
    \\[5pt]
     & all-to-all & $\order{1}$, $\order{N}$ & $\order{\log N}$, $\order{N}$ & $\order{1}$ & $\order{ \log N}$
    \\[5pt]
    \makecell[l]{RZ per gate} && 1 & 1 & 4 to 8 & 1
    \\\bottomrule
    \end{tabular}
    \caption{\textbf{Summary of the comparison between iQCC, vDBF, SWCS, and ADAPT-VMPE iterative algorithms.} To our knowledge, ADAPT-VMPE is the only method that attains all the positive properties listed above, which impacts convergence rate and the number of CNOTs or T gates. In particular, it is the only algorithm for which the original energy landscape is continuous, and that uses active rotations simultaneously. The only property of ADAPT-VMPE not presented in this paper is pool synthesis, which is currently under investigation (initial analysis suggests a performance gain with no quality drop). For all-to-all connectivity, ADAPT-VMPE is only logarithmically far from SWCS and iQCC, which in our opinion is negligible given the improvement in convergence. The continuous parameter landscape for iQCC was achieved by building a further surrogate function on top of the one constructed by coefficient truncation; see Ref.~\cite{ryabinkin2025optimization} for details. See Appendix~\ref{appendix:transpilation-cost-analysis} for CNOT and RZ cost analysis for all considered methods. Both Fermionic-local~\cite{mullinax2025classical} and qubit-local~\cite{larose2026cost} pools have been considered for SWCS.}
    \label{tab:summary}
\end{table}

\subsection{Large-scale benchmarks of ADAPT-VMPE}

We present large-scale results for Complex A~\cite{McFarland2019,tarocco2025aegiss}. We used Heisenberg ADAPT-VMPE, in which the new gate was always added at the beginning, appending to our initial state corresponding to the solution of the Hartree-Fock equations. We considered active space sizes resulting in 28, 40, 52, 64, 76, 88, and 100 qubits. For all active space sizes considered, we ran the ADAPT-VMPE algorithm for 200 iterations for 100 qubits, 300 iterations for 88 qubits, and 400 iterations for all other active spaces. As a reference, we employed the Density Matrix Renormalization Group (DMRG) method~\cite{White1992, DMRG_CHEM, DMRG_CHEM_2}. Using an adequate bond dimension, DMRG yields energies close to CASCI/FCI quality. Such methods quickly become unfeasible for large active spaces. DMRG is therefore a suitable classical benchmark for assessing the convergence behavior of our ADAPT-VMPE simulations. For additional details on the computational setup, see Appendix~\ref{sec:molecules}.
As shown in Fig.~\ref{fig:tld-large}, the energy converges systematically toward the reference energies, confirming the effectiveness of the algorithm. Some performance improvements are applied with a negligible impact on quality to speed up the simulations. The cumulative time for generating the results is presented in Fig.~\ref{fig:tld-large}. From the plot, we see that the 100-qubit ansatz with $\approx 100$ gates can be generated in $\approx 100$ hours. From the plot, we can confirm that ADAPT-VMPE indeed scales polynomially with the number of qubits. 

\begin{figure}
\includegraphics[]{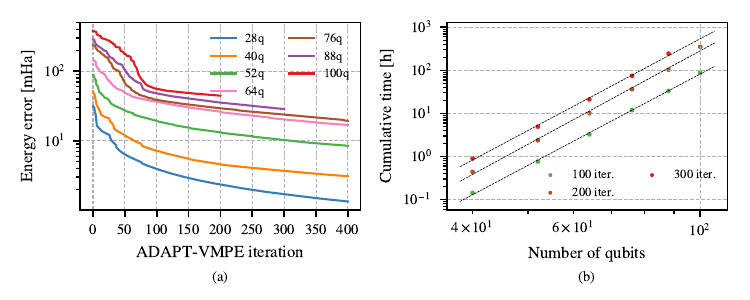}
\caption{\textbf{Heisenberg ADAPT-VMPE simulations performed for active space sizes ranging between 28 and 100 qubits for Complex A.} The MP cutoff is set to 6 for all simulations. A pool reduction strategy is applied, combining GGF selection with faster gradient selection as described in the Methods section. Active rotations are included in the ansatz, and the whole ansatz is optimized after adding each gate using the L-BFGS-B optimizer. In (a), the energy error is plotted with respect to the energy computed by DMRG. In (b), the scaling of computation time with the number of qubits is presented. The cumulative time scales polynomially for all ADAPT iteration counts considered.}
\label{fig:tld-large}
\end{figure}

The ADAPT-VMPE algorithm is executed with cutoff $c=6$ for Majorana Propagation. We verified the ans\"atze at iterations 100 and 300 (except for 88 and 100 qubits) for each active space size with higher cutoffs $c=8$ and $c=10$. The Table~\ref{tab:diff-cutoff} presents the differences between cutoffs 6 and 8, and between 8 and 10. The latter typically gives values an order of magnitude smaller than the former. In addition, the differences between consecutive cutoffs are positive (with the exception of the 76-qubit case at 100 iterations), suggesting that the energy computed during ADAPT-VMPE was in fact mostly \textit{upper bounding} the true energy of the ansatz. This effect is consistent across all considered active space sizes and confirms the high quality of the produced circuits.

\begin{table}[]
\begin{tabular}{@{}lllllllll@{}}
\toprule
 &  & \multicolumn{7}{c}{Qubits} \\ \midrule
 &  & 28 & 40 & 52 & 64 & 76 & 88 & 100 \\ \midrule
\multirow{2}{*}{100 iter.} & c6-c8 & 0.0178 & 0.0181 & 0.0554 & 0.3125 & $-$0.1778 & 0.4159 & 1.4753 \\
 & c8-c10 & 0.0012 & 0.0008 & 0.0028 & 0.0183 & 0.0895 & 0.0392 & 0.0574 \\
 \addlinespace[3pt]
\multirow{2}{*}{300 iter.} & c6-c8 & 0.0274 & 0.0420 & 0.0044 & 0.2819 & 0.1028 & 0.0005 & --- \\
 & c8-c10 & 0.0013 & 0.0008 & 0.0039 & 0.0251 & 0.0879 & 0.0006 & --- \\ \bottomrule
\end{tabular}
\caption{\label{tab:diff-cutoff} \textbf{Differences between the energies computed with MP at different cutoffs, in mHa.} c6-c8 and c8-c10 denote the difference between cutoffs 6 and 8 and between cutoffs 8 and 10, respectively. We analyzed ADAPT-VMPE circuits with 100 and 300 gates for various active space sizes for Complex A. The difference between cutoffs 8 and 10 is typically an order of magnitude smaller than between cutoffs 8 and 6. To speed up energy computation, in addition to the standard MP truncation, we employed coefficient-magnitude truncation with threshold $\tau=10^{-8}$.}
\end{table}

Whether in the NISQ or FT era, it is vital to decrease the overall cost of running the circuits. For the NISQ era, a typical choice is to limit the number of two-qubit gates, which are the main source of errors. As the constructed circuit consists of Fermionic operators, we need to select a suitable F2Q mapping and transpile the qubit circuit on the machine. For this purpose, we use the treespilation procedure described in~\cite{miller2024treespilation} for both all-to-all and heavy-hex topologies. Before transpiling the ansatz, we remove the active rotations and dress them to the Hamiltonian, transforming it into a different second-quantized Hamiltonian, which allows further resource savings. Results are presented in Fig.~\ref{fig:compilation}, and we can see that for all active space sizes, even the most complex circuits do not exceed $\approx 8500$ CNOTs for both topologies. These numbers are of the order of what was recently achieved by the IBM Quantum Nighthawk processor~\cite{ibmDeliversQuantum}. For the FT era, the typical choice is to minimize the non-Clifford T gates. The cost is estimated by transpiling all parameters using the Ross-Selinger method, which provides the optimal T count in the worst case scenario~\cite{ross2014optimal}. In Fig.~\ref{fig:compilation}, we present the T count for the previously produced results for Complex A. The T count for the most complex 100-qubit ansatz is at most $28,\!000$, which is an order of magnitude smaller than what is typically needed for QPE~\cite{casares2022tfermion}. We also dress the active rotations to the Hamiltonian in the FT setting, as the modified Hamiltonian is simpler than the original one, as shown in Appendix~\ref{appendix:compilation}.

\begin{figure}
\includegraphics[width=\linewidth]{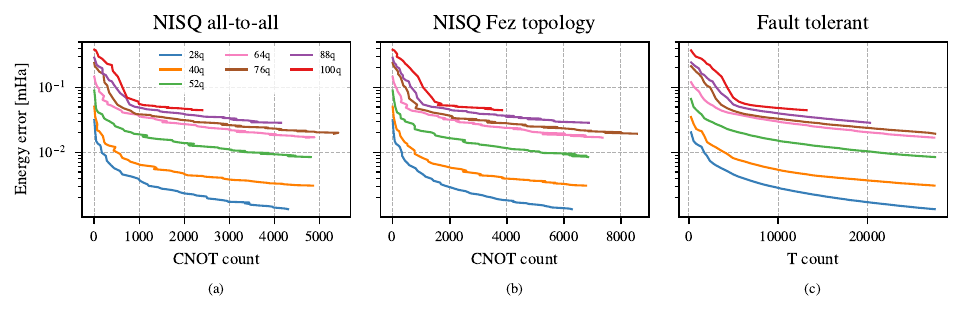}
\caption{\textbf{Energy error vs resource cost of ADAPT-VMPE outcomes transpiled for NISQ and FT eras.} (a) all-to-all connectivity minimizing 2-qubit gates for the NISQ era, (b) \texttt{ibm\_fez} heavy-hexagonal connectivity minimizing 2-qubit gates for the NISQ era, and (c) fault-tolerant scenario minimizing T count.}
\label{fig:compilation}
\end{figure}

Like the typical VQE method, our approach is designed for minimizing the energy $E\coloneqq \bra{\psi}H\ket{\psi}$ for a given Hamiltonian $H$. However, the overlap of the constructed state with the ground state also has significant implications, in particular for algorithms like QPE. When searching for the ground state, convergence of the energy $E$ to the ground state energy via the variational principle implies that the overlap between $\ket{\psi}$ and the ground state approaches 1. This holds when methods like QPE can reach chemical precision on their own. However, when chemical precision is not achieved, the energy alone may not be sufficient to deduce the overlap itself. Therefore, computing the overlap with the target state may give additional insight into the state obtained. In Fig.~\ref{fig:vmpe-vs-dmrg}, we show the overlap analysis of the produced states computed using the formula described in the Methods section. As shown in the left plot, the lower bound on the overlap increases monotonically towards the optimal value of 1. We also show that the total number of degrees of freedom scales polynomially to reach a constant lower bound on the overlap.

\begin{figure}
\includegraphics[width=\textwidth]{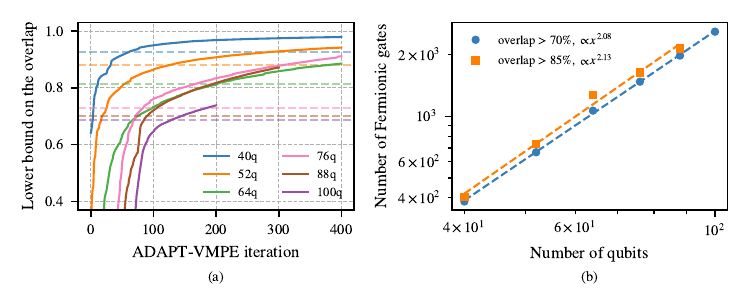}
\caption{\textbf{Analysis of overlap of the produced states with respect to the ground state for Complex A.} Subfigure (a) shows that the lower bounds on the overlap increase monotonically with iterations. For reference, the dashed lines represent the exact overlaps of the Hartree-Fock determinant and the ground state computed with DMRG. In (b), we compare the scaling of the number of Fermionic gates for ADAPT-VMPE to reach 70\% and 85\% overlap for various active spaces.}
\label{fig:vmpe-vs-dmrg}
\end{figure}

\section{Discussion}\label{sec:discussion}

We presented the ADAPT-VMPE algorithm, an entirely classical algorithm that constructs a Fermionic circuit approximating the ground state of a second-quantized operator. Using Majorana Propagation, the algorithm can run using only classical resources, converging towards the target ground state in polynomial time. Our method does not require a GPU, which makes it a particularly inexpensive procedure to run. Even for the 100-qubit example of Complex A, our algorithm can generate a circuit with $\approx$75\% overlap with the target state in only $\approx$350 hours, within only $200$ iterations.

ADAPT-VMPE is not merely a combination of techniques that happen to fit each other: we have analyzed each aspect (parameter optimization, evolution direction, pool gate selection) to derive a combination that yields a very high quality algorithm for state preparation, while keeping in mind recent F2Q development and transpilation techniques, whether for 2D connectivity like heavy-hexagonal or all-to-all connectivity for scaling analysis. Our algorithm is a natural fit for Majorana Propagation, which additionally allows high-quality expectation value approximation for the circuits we use, providing a smooth landscape for optimization. 
The produced circuits are agnostic to the architecture, allowing selection of the compilation method and/or quantum computer at a later stage. Since the output is defined in the Fermionic space, one can adjust the F2Q mapping at any point, using, e.g., treespilation~\cite{miller2024treespilation} to tailor it to the quantum device. We showed that even for the largest active space size, the number of 2-qubit operations is of the order of those currently runnable on quantum computers~\cite{ibmDeliversQuantum}. The Fermionic circuits can also be compiled for fault-tolerant machines. The most complex circuits produced for 100 qubits require at most $\approx250,\!000$ T gates if active rotations are included in transpilation, or $28,\!000$ otherwise. Both figures are orders of magnitude smaller than what is typically needed for QPE. For comparison, benchmarks of various QPE implementations indicate that for a $\approx50$ qubit problem, QPE already requires $\approx 5.3\cdot 10^{12}$ T gates~\cite{casares2022tfermion}. Recent studies suggest that for phase estimation, properly transformed Hamiltonians allow further reducing the number of Toffoli gates (each implementable with 7 T gates and some Clifford gates) to $10^8$ for $\approx$100 qubit Hamiltonians~\cite{low2025fast}. Since the cost of implementing ADAPT-VMPE compared to QPE itself is negligible, we claim the circuits produced by ADAPT-VMPE greatly improve upon at least the naive choice of the HF state as an initial state for QPE.

We would like to stress that other methods exist for state preparation using only classical resources. In particular, one can mention synthesizing Matrix Product States and Sum of Slater determinants into circuits. Both of these methods produce a wavefunction in a form not directly runnable on a quantum computer, typically using DMRG~\cite{DMRG_CHEM} or the SCI algorithm~\cite{ogSCI, cipsi}. Comparing against these algorithms is an interesting direction for future research.

\section{Methods}

\subsection{Recap of Majorana Propagation}
Inspired by prior work on Pauli Propagation~\cite{rall2019simulation, rudolph2025pauli}, Majorana Propagation (MP) was introduced in~\cite{miller2025simulation} as an algorithmic framework to efficiently estimate expectation values of Fermionic observables propagated through Fermionic circuits and has since then been extensively used and extended in the literature \cite{danna2025majorana,alam2025fermionic, bako2025fermionic, facelli2026fast,rudolph2026thermal}. The algorithm proceeds by applying successive, controlled truncations during the Heisenberg evolution of the observable.
Its core insight is to use ``monomial length'' as a principled truncation criterion.
For unstructured circuits, high-length Majorana monomials are exponentially unlikely to contribute to expectation values, and backflow from high-length to lower-length monomials is quadratically suppressed.
These properties yield concrete performance guarantees: the approximation error decays exponentially with the truncation threshold, and computing expectation values to a fixed error for typical circuit ensembles requires only polynomial resources.

To ensure we are in the regime where MP's performance guarantees apply, we impose the following conditions on the circuit and observable.
For the circuit, we consider instances being products of exponentiations of monomial generators $M_{\gamma_j}$ drawn uniformly at random from the set of length-4 monomials, with angles independently sampled from $[0,2\pi]$.
For the observable, we assume a randomly chosen second-quantized molecular Hamiltonian.
While these assumptions may appear restrictive and not tailored to brickwork-style circuits such as tUPS~\cite{burton2024accurate} or Jastrow~\cite{motta2023bridging} ans\"atze, they accurately model random compilation settings such as qDRIFT~\cite{campbell2018random}.
Notably, ADAPT-VQE aligns closely with this regime: the algorithm effectively produces an unstructured circuit whose generator pool is dominated by length-4 monomials, and the target observables are unstructured second-quantized Hamiltonians subject to symmetries.

Concretely, given an observable $H=\sum_{\nu} c_\nu M_{\nu}$ expressed in the Majorana basis, an initial Fock basis state $\varrho=\ket{n_1\dots n_n} \bra{n_1\dots n_n}$, and a unitary Fermionic circuit $C_L(\vec{\theta}\,)$, the MP algorithm approximates the expectation value
\begin{equation}
    \tilde f_L(\vec{\theta}\,) \approx f_L(\vec{\theta}\,) = \langle H \rangle_{\vec{\theta}} = \Tr[\varrho \, C_L(\vec{\theta}\,)^{\dagger}H C_L(\vec{\theta}\,)],
\end{equation}
where
\begin{equation}
    \label{eq: UL}
    C_L(\vec{\theta}\,) = \prod_{j=1}^L U_j = \prod_{j=1}^L e^{-\text{i} \theta_j M_{\gamma_j}/2}.
\end{equation}
The unitary $C_L$ consists of $L$ Fermionic operators, where each gate $U_j$ is generated by a Majorana monomial $M_{\gamma_j}$ and $\theta_j\in\mathbb{R}$ is a gate parameter.
Here we use the shorthand notation of greek letters to represent the bitstrings that specify hermitian monomials, e.g., $M_\nu = \mathrm{i}^{|\nu|(|\nu| - 1)/2} m_1^{\nu_1} m_2^{\nu_2} \cdots m_{2N}^{\nu_{2N}}$ where $\nu \in \{0,1\}^{2N}$, where $|\nu| \coloneqq \sum_i \nu_i$ is Majorana monomial length.
The MP algorithm back-propagates each monomial in $H$ through the gates of $C_L$.
For each monomial in $H$ there are three scenarios for the propagation through a gate:
\begin{equation}
    M_\nu \xrightarrow{U_j}
    \begin{cases}
        M_\nu                                                                      & \text{if } [M_\nu, M_{\gamma_j}] = 0,                             \\
        \cos(\theta_j)\, M_\nu + \mathrm{i} \, \sin(\theta_j)\, M_{\gamma_j} M_\nu & \text{if } \{M_\nu, M_{\gamma_j}\} = 0, \\
        \cos(\theta_j)\, M_\nu                                                     & \text{if } \{M_\nu, M_{\gamma_j}\} = 0 \text{ and truncated}.
    \end{cases}\label{eq:mp-dispatch}
\end{equation}
In other words, if the propagated monomial $M_\nu$ commutes with $M_{\gamma_j}$, then $M_\nu$ is unchanged.
Otherwise, $M_\nu$ branches and produces a new monomial $\propto M_{\gamma_j} M_\nu$, and if the resulting monomial exceeds a truncation threshold, it is discarded.
At each stage, we update the observable iteratively, starting with $H_0 = H$ by applying the transformation $H_{j+1} = e^{-\mathrm{i} \, \theta_{L-j} M_{\gamma_{L-j}}/2} H_j e^{\mathrm{i} \, \theta_{L-j} M_{\gamma_{L-j}}/2}$,
where we collect the newly generated monomials to construct the evolved observable $H_{j+1}$.
After propagating all monomials in $H$ through the circuit $C_L$, we compute the expectation value by evaluating the surviving monomials on the initial state $\varrho$.
The key insight is to express a Fock-basis state in the Majorana basis, which reveals the structure of non-zero overlaps.
We express the Fock-basis projector $\ketbra{n_1 \dots n_N}$ in the Majorana basis as
\begin{equation}
    \ketbra{n_1 \dots n_N}
    = \frac{1}{2^{N}}\Big(1+
    \sum_{i\leq 1}(-1)^{n_i} \overline{m}_{i}+
    \sum_{1 \leq i < j \leq N}(-1)^{n_{i} + n_{j}} \overline{m}_{i}\overline{m}_{j}
    +\cdots\Big),
\end{equation}
where $\overline{m}_{j} = -\mathrm{i} m_{2j-1} m_{2j}$ represents the paired Majorana operators for orbital $j$.
We refer to products of these paired operators, $\prod_j\overline{m}_{j}$, as \textit{paired monomials}, which form the natural basis for expressing Fock-basis states.
Given the evolved observable $H_L=\sum_{\nu} c_\nu M_{\nu}$ and Fock state $\varrho  = \sum_\mu b_\mu M_{\mu}$, the expectation value becomes
\begin{equation}
    \begin{aligned}
        \tilde f_L(\vec{\theta}) & = \Tr[\varrho H_L]                                                   = \frac{1}{2^N} \sum_{\nu, \mu} c_\nu b_\mu \, \Tr[M_{\mu} M_{\nu}]  = \sum_{\nu} c_\nu b_\nu,
    \end{aligned}
\end{equation}
where we have used the orthogonality relation for hermitian Majorana monomials: $\Tr[M_{\nu} M_{\mu}] = 2^N \delta_{\nu, \mu}$.
This simplification shows that only monomials present in both the evolved observable and the initial Fock-basis state contribute to the expectation value, effectively filtering out all other terms.
Thus, rather than computing a sum over the exponential number of terms in the state, we need only consider the paired monomials that overlap with the state.

Although operator propagation methods are usually presented in the Heisenberg picture, an equivalent formulation is obtained by forward-propagating the gates through the state and directly computing the overlap with the Hamiltonian.
In Appendix~\ref{app:schrodinger equivalence}, we show that, under monomial-length truncation, the expectation values produced by Heisenberg- and Schrödinger-picture evolution are identical. Consequently, all performance guarantees remain unchanged. 
This equivalence is useful, as it allows gates to be appended to either the front or the back of the circuit without affecting the approximation error.

The energy landscape with respect to each parameter $\theta_j$ remains of the form $A\sin(x+B)+C$, as in the case of the original gate $U_j$~\cite{ostaszewski2021structure}. This is important as it allows use of the GGF formulas for gate selection when Majoranic pool is used~\cite{feniou2025greedy}. In addition, since truncating based on monomial length does not change the branching structure, one can create a fixed \textit{surrogate graph} that, once constructed, allows fast repetitive evaluation of $\tilde f$ and its gradient for different parameters. 
This graph can be constructed in both Heisenberg and Schr\"odinger pictures with no impact on the approximation error. One can also modify the truncation techniques to improve the approximation and/or speed up the surrogate graph construction. Further details and an intuitive explanation of why Majorana Propagation works are reviewed in Appendix~\ref{sec:majoranaprop}.

\subsection{Choice of the pool}
The choice of pool has a significant impact on ADAPT-VQE and ADAPT-VMPE procedures, affecting both convergence and the final resource cost. Existing pools can be split into two categories: those acting locally in the Fermionic space, such as Fermionic~\cite{grimsley2019adaptive}, Fermionic-spin~\cite{tsuchimochi2022adaptive}, or Majoranic~\cite{miller2024treespilation} pools, and those acting locally in the qubit space, such as QEB~\cite{yordanov2021qubit} or qubit-ADAPT~\cite{tang2021qubit}, or CEO~\cite{ramoa2025reducing} pools. As shown in the Results section, the former category yields better convergence, which is intuitively expected as these pools naturally inhabit the Fermionic space. However, when Jordan-Wigner mapping is used, Fermionic pool elements are on average $\order{N}$-local in qubit space. The qubit-local pools were proposed as a remedy, as each gate is at most a 4-local qubit operator.

With recent circuit compilation techniques, this discrepancy no longer holds. For all-to-all connectivity, one can select a ternary-tree mapping~\cite{jiang2020optimal} that guarantees a Fermionic gate is at most $\order{\log N}$-local in qubit space, giving negligible overhead over the bounded locality of qubit-local pools. For 2D-like structures such as grids or heavy-hexagonal lattices, the four qubits acted on by a QEB or qubit-ADAPT operator are typically separated by a distance of $\order{\sqrt N}$ from each other, resulting in an $\order{\sqrt N}$ CNOT cost per gate. For Fermionic-local pools, the same scaling is achieved with a properly selected Bonsai mapping~\cite{miller2023bonsai}, and can be further optimized with the treespilation procedure, making the CNOT cost strictly below that of QEB or qubit-ADAPT pools~\cite{miller2024treespilation}. In conclusion, qubit-local pools offer no meaningful per-gate advantage over Fermionic-local pools, which combined with better ADAPT-VQE convergence motivates use of the latter. 

\subsection{Complexity analysis of ADAPT-VMPE}

ADAPT-VMPE provably runs in polynomial time in both the number of qubits and the number of iterations $K$. The key components, and at the same time the most time-consuming steps, are preprocessing of the ansatz into the Majorana Propagation framework to allow fast energy evaluation, evaluating pool elements to select the most effective gate, and parameter optimization. Assuming Majorana Propagation with fixed cutoff $c\geq 6$, the preprocessing of the ansatz takes $\order{N^{c}K}$ time. The same complexity holds whether a surrogate graph is constructed in the Heisenberg or Schr\"odinger picture. For pool selection, evaluating each gate requires $\order{N^{c}K}$ time, resulting in a total time complexity of $\order{N^{c+4}K}$ since the pool size is $\order{N^4}$. For parameter optimization, single energy and gradient evaluations using the technique presented in Appendix~\ref{sec:majoranaprop} each take $\order{N^{c}K}$, to be multiplied by the number of optimization steps. Since the total number of optimization steps is not expected to grow as $\order{N^{5}}$, in the worst case the algorithm scales as $\order{N^{c+4}K}$, which is polynomial in $N$ and linear in $K$ for fixed cutoff $c\geq 6$. In practice, the scaling is much more favourable in $N$; in particular, our analysis of Fig.~\ref{fig:tld-large} shows that the time scales as $\order{N^{7.07}}$ for 300 iterations.

\subsection{Details on large-scale simulations}

To generate circuits for Complex A, we used ADAPT-VMPE as described above. We used Majorana Propagation with cutoff $c=6$, with the original length truncation, and with paired acceptance modification. We used the L-BFGS-B algorithm for parameter optimization. New gates were always added to the beginning of the ansatz, right after the Hartree-Fock determinant implementation.

\paragraph{Performance improvements}
Although ADAPT-VMPE is provably polynomial, the runtime can be further accelerated with only a negligible impact on quality. First, one can show that for a Majoranic gate from the pool, exchanging a Majorana in the monomial with another one acting on the same spin-orbitals has no effect on the evolved state up to the sign of the parameter, when the gate acts on a computational basis state, an effect similar to the one discovered in~\cite{ryabinkin2020iterative}. This allows reducing the pool size by 8 in Heisenberg ADAPT-VMPE, since new gates are always added at the beginning of the circuit. Furthermore, we observed that for both gradient and GGF selection, many highly ranked candidate gates at a given iteration tend to remain highly ranked in subsequent iterations. Taking advantage of this, we introduced a pool trimming strategy in which the full pool is evaluated only at selected iterations, while a smaller subset of $5\%$ (for 28 qubits) and $1\%$ (for other active space sizes) of the full pool is used for 25 steps, after which pool trimming is repeated on the full pool. Since our implementation of gradient selection is much faster but less accurate than GGF selection, we use the former for trimming the pool and the latter for selecting the gate to be added in intermediate steps. Further explanation and investigations can be found in Appendix~\ref{sec:adapt-vmpe-performance}.

\paragraph{Transpilation} ADAPT-VMPE produces circuits of the form described in Eq.~\eqref{eq:adapt-vmpe-output}, consisting of the Hartree-Fock ansatz, Fermionic gates, and active rotations. In principle, we only need to synthesize the Hartree-Fock gates and the gates added during ADAPT-VMPE iterations. The active rotations can be implemented in the Heisenberg picture on the observable, transforming it into a different second-quantized operator. To compile the remaining part of the circuit for a quantum computer, we first select a suitable F2Q mapping and then transpile the qubit circuit on the machine. Below we describe the compilation procedure used to generate the data for Fig.~\ref{fig:compilation}; remaining technical details can be found in Appendix~\ref{appendix:compilation}.

For NISQ transpilation, we used the treespilation procedure~\cite{miller2024treespilation}, which optimizes the CNOT count over tree-based F2Q mappings. Since the Ternary Tree mapping is a special case of the Bonsai mapping, it guarantees implementing ADAPT-VMPE circuits in $\order{K \log N}$ CNOTs for all-to-all connectivity. For 2D structures such as grids or heavy-hexagonal lattices, optimization is performed over Bonsai mappings whose underlying ternary tree is a subgraph of the heavy-hexagonal graph. In this setting, circuits can be implemented using $\order{K \sqrt N}$ CNOTs, a scaling related to the height of the ternary tree~\cite{miller2023bonsai}. 

For the FT era, the goal is to minimize the non-Clifford T gates. Since any $\exp(i \theta P)$ for a Pauli operator $P$ can be transformed into a $Z$-rotation via Clifford operations, the T-count depends only on the parameter $\theta$ and not on $P$, meaning the transpilation cost is independent of the selected F2Q mapping. We used the Ross-Selinger method~\cite{ross2014optimal} to transpile all $Z$ rotations in the circuit. 

\paragraph{Overlap analysis}

While Majorana Propagation allows faithful energy estimation, it does not allow direct computation of the overlap, since Fock states are linear combinations of exponentially many Majorana monomials. For this reason, we estimate lower bounds on the overlap of the produced state with the ground state of Hamiltonian $H$ using the formula  
\begin{equation}\label{eq:overlap-bound}
    |\braket{\psi}{ E_0}|^2  \geq \begin{cases}
     \frac{S_1-\bra \psi H \ket{\psi}}{S_1-E_0} - \frac{p}{\lambda_2}\frac{S_1-S^\top_1}{S_1-E_0}, &S^\top_1 <S_1,\\
     \frac{S_1-E}{S_1-E_0} - \frac{p}{\lambda_P} \frac{S_1-S^\top_1}{S_1-E_0}, & \rm otherwise,
    \end{cases}
\end{equation}
where $E_0$ is the approximate ground energy, $S_1$ is the approximate first excited singlet energy from the correct particle sector, $S_1^\top$ is the approximate first excited non-singlet energy, and for $H_P$ being a nonnegative observable whose null space is spanned by singlet states of the correct sector, $p\coloneqq \bra \psi H_P \ket{\psi}$, and $\lambda_2$ ($\lambda_P$) are the smallest (largest) positive eigenvalues. A proof of the formula can be found in Appendix~\ref{appendix:overlaps}. In general, $H_P$ takes the form
\begin{equation}
    H_P = A (\hat N - N_{\rm exp})^2  + B \hat S_z ^2 + C \hat S^2,
\end{equation}
where $A,B,C>0$ are penalty constants, $\hat N$ is the number operator, $N_{\rm exp}$ is the number of particles in the system, $\hat S_z$ is the $z$-spin operator, and $\hat S^2$ is the total spin operator. One can set $B=0$ since breaking $z$-spin symmetry already implies breaking total spin $\hat S^2$ for singlet states. For our calculations, we set $A=1$, $B=0$, and $C=\frac{4}{3}$, giving $\lambda_2 = 1$ and $\lambda_P = \max(N_{qubits} - N_{exp}, N_{exp})^2 + \frac{2}{3}N$. Bounds on the eigenvalues are sufficient in place of their exact values. Values for $E_0$, $S_1$, and $S_1^\top$ were approximated with DMRG calculations and are listed in Appendix~\ref{sec:molecules}. The quantities $p$ and $\bra \psi H \ket \psi$ were computed with MP at cutoff 6.

The HF contribution is estimated by examining the weight of the corresponding configuration state function/Slater determinants in the reference DMRG MPS. All information related to the DMRG setup and calculations is reported in Appendix~\ref{ssec:tld_dmrg}. 

\backmatter

\bmhead{Acknowledgments} Work on ``Quantum Computing for Photon-Drug Interactions in Cancer Prevention and Treatment'' is supported by Wellcome Leap as part of the Q4Bio Program. We acknowledge EuroHPC Joint Undertaking for awarding the projects EHPC-REG-2025R05-171 and EHPC-DEV-2025D30-140 access to Leonardo at CINECA, Italy. 

The authors would like to thank Stefan Knecht, Zolt\'an Zimbor\'as, and Guillermo Garc\'ia P\'erez for valuable discussions and comments on the manuscript, and Ludmila Botelho and Roberto Di Remigio Eik\aa s for assistance with the coding process.

\bmhead{Author contributions} \"OS, AN, AM and AG conceived the ADAPT-VMPE method. RC, \"OS, AN, AM and AG implemented the ADAPT-VMPE method. RC, \"OS, AN, and AG performed the numerical computations related to ADAPT-VMPE. RC, FP, PH, MS, FT prepared Hamiltonians for the TLD-1433 derivative and analyzed its spectrum. AG directed the work. All authors participated in discussions and in the writing of the manuscript. 

\bmhead{Competing interests} Elements of this work are included in patent applications filed by Algorithmiq Oy currently pending with the European Patent Office and/or the United States Patent and Trademark Office.

\bmhead{Data and code availability} 
The data that support the findings of this article are available on Zenodo under~\cite{zenodo}. The code used to obtain these data cannot be made publicly available because it contains commercially sensitive information.

\bibliography{mp-adapt-vqe.bib}

\appendix 
\section{Molecules description}\label{sec:molecules} 

\subsection{H-chain} Hydrogen chains are widely used as benchmark systems because they combine strong, tunable electron correlation with a simple, scalable geometry~\cite{chan2002, Motta2017, Motta_2020}, and are generally used for benchmarking quantum algorithms for state preparation. The most common form of a hydrogen chain is defined by selecting the desired number of hydrogen atoms and placing them equidistantly in a linear configuration with bond distance $r$. In standard settings, the number of electrons equals the number of hydrogen atoms. In our case, we selected a chain of eight hydrogen atoms separated by a bond distance $r=1.5$ Å. The initial SCF calculation was performed with the cc-pVTZ~\cite{ccpvtz} basis set. To recover additional correlation, we performed an MP2~\cite{mp2} calculation on top of the SCF one. From the obtained wavefunction, we built the 1-RDM and computed the natural orbitals, restricted to the virtual space, defining a set of Frozen Natural Orbitals (FNOs). The active space is then defined by combining canonical HF orbitals for the occupied space with the leading FNOs for the virtual space. In total, the active space is composed of 8 electrons in 8 orbitals, i.e., 16 qubits. Additionally, the particle number operator was added as a penalty with a coefficient of 0.1. The reference energy, obtained with a CASCI calculation, is $-9.057128$ Ha.

\subsection{Complex A}\label{ssec:tld_dmrg}
Complex A used in this work is TLD-1411 (rac-[Ru(bpy)$_2$(IP-3T)]Cl$_2$)~\cite{McFarland2019}, employed in place of TLD-1433 for computational practicality. TLD-1411 is a ruthenium(II) polypyridyl complex and a very close structural analogue of TLD-1433 (rac-[Ru(dmb)$_2$(IP-3T)]Cl$_2$), a drug currently undergoing human clinical trials for the treatment of non-muscle-invasive bladder cancer. The two compounds share an almost identical molecular scaffold and highly similar electronic structures. The only structural difference lies in the ancillary ligands: TLD-1411 contains bpy (2,2'-bipyridine) ligands, while TLD-1433 contains dmb (4,4'-dimethyl-2,2'-bipyridine) ligands.

The presence of methyl substituents in the dmb ligand introduces minor electronic and steric effects, leading to very similar absorption properties, excited-state energetics, and photochemical pathways relevant to their therapeutic activity. The qualitative trends obtained for TLD-1411 are therefore expected to be directly transferable to TLD-1433, without affecting the generality of the conclusions.

The geometry of TLD-1411 was optimized in water using the Conductor-like Polarizable Continuum Model (C-PCM)~\cite{TRUONG1995253, Barone1998}. Calculations employed the PBE0 functional~\cite{PBE0_1, PBE0_2} together with the def2-TZVP basis set~\cite{def2tzvp} in the Orca code~\cite{ORCA}. The identification of the relevant occupied molecular orbitals used to construct the active spaces required for a reliable description of the excitations investigated in this work was carried out following the procedure outlined in Ref.~\cite{tarocco2025aegiss}. The resulting active spaces include orbitals originating from two main regions of interest within the molecule: the ruthenium center and the elongated thiophene tail (further details can be found in Ref.~\cite{tarocco2025aegiss}). For the occupied space, we employed canonical Unrestricted Hartree-Fock orbitals obtained from a $\Delta$-SCF~\cite{gill2008} calculation using the def2-SVDP~\cite{def2svpd} basis set in water and the cc-PVDZ-pp pseudopotential~\cite{pp_transcl}. The virtual space was defined using Frozen Natural Orbitals generated from an unrestricted MP2 calculation, tailored to the specific canonical occupied orbitals. In particular, we selected active space sizes composed of 14, 20, \ldots, 50 orbitals, resulting in 28, 40, \ldots, 100 qubits, with 14, 20, \ldots, 50 particles.

All DMRG reference energies and approximated eigenstate estimates, reported in Tab.~\ref{tab:tld-1433-ref-s0} and Tab.~\ref{tab:tld-1433-overlap-eigenvalues} respectively, are obtained using the \texttt{Block2}~\cite{block2_1} implementation of the algorithm. Given the unrestricted nature of the orbitals, the non-spin-adapted DMRG mode \texttt{SymmetryType.SZ} was used for all calculations. We stress that the reference energies are approximations of the true ground state given a specific bond dimension. The true ground state may lie slightly lower than the reported reference energies. Energies in Tab.~\ref{tab:tld-1433-ref-s0} are computed using state-specific DMRG targeting the lowest singlet state, imposing \texttt{spin=0}. 

To compute the lower bound on the overlap as described in Appendix~\ref{appendix:overlaps}, we require approximations of several eigenstates of the Hamiltonians. The values reported in Tab.~\ref{tab:tld-1433-overlap-eigenvalues} are obtained in a mixed fashion. The two singlet roots S$_0$ and S$_1$ are obtained by performing a two-root state-averaged DMRG calculation followed by refinement with state-specific DMRGs using an orthogonality projector (to prevent the second root from collapsing onto S$_0$). Two further eigenstates have been computed: the lowest triplet state $T_1$ and the lowest quintet state $Q$. For these, we run two independent state-specific DMRG calculations targeting the lowest solution with \texttt{spin=2} for T$_1$ and \texttt{spin=4} for the quintet $Q$. 

Since Majorana monomials can break all number, $z$-spin, and $S^2$ symmetries, $S_1^\top$ was selected to be the smaller of the T$_1$ and $Q$ state energies. Other possible candidates for $S_1^\top$ could include doublet states with one extra or one missing electron; however, even-length Majorana monomials from the Majoranic pool~\cite{miller2024treespilation} construct states that are superpositions of Fock states of the same parity as the Hartree-Fock state, so we omit these states from our analysis. Given that Ru(II) polypyridyl complexes are low-spin d$^{6}$ closed-shell systems, their lowest excitations correspond to single-electron transitions of singlet and triplet spin multiplicities. Other multiplicities would require higher excitation manifolds and lie significantly higher in energy. We verified this computationally, confirming that no low-lying quintet or septet states are present.

\begin{table}[]
    \centering
    \begin{tabular}{@{}c@{\hspace{2cm}}c@{\hspace{2cm}}c@{}}
    \toprule
    Active space & $\chi$ & S$_0$ \\\midrule
    14e14o (28q) & 2000 & $-3441.833359$  \\
    20e20o (40q) & 2000 & $-3441.852417$  \\
    26e26o (52q) & 2000 & $-3441.891128$  \\
    32e32o (64q) & 2500 & $-3441.947764$  \\
    38e38o (76q) & 3000 & $-3442.039072$  \\
    44e44o (88q) & 3000 & $-3442.087434$  \\
    50e50o (100q) & 4000 & $-3442.181423^{\dagger}$ \\
    \bottomrule
    \end{tabular}
    \caption{\textbf{Reference DMRG energies (in Ha) for S$_0$ calculated with different bond dimensions ($\chi$) for different active spaces.} Active spaces are denoted as number of electrons, number of spatial orbitals, and number of qubits in parentheses. ($\dagger$): Energy reference obtained using optimal orbital reordering~\cite{DMRG_CHEM_2}, allowing escape from local minima. }
    \label{tab:tld-1433-ref-s0}
\end{table}

\begin{table}[]
    \centering
    \begin{tabular}{@{}cccccc@{}}
    \toprule
    Active space & $\chi$ & S$_0$ & S$_1$ & T$_1$ & Q \\\midrule
      14e14o (28q)& 750 & $-3441.833345$ & $-3441.629474$ & 	$-3441.629250$	& $-3441.421192$ \\
    20e20o (40q) & 1000 & $-3441.852286$ & $-3441.711514$ & $-3441.711676$ & $-3441.508494$ \\
    26e26o (52q) & 1500 & $-3441.890803$ & $-3441.750981$ & $-3441.751480$ & $-3441.549015$ \\
    32e32o (64q) & 1500 & $-3441.945401$ & $-3441.818441$ & $-3441.820526$ & $-3441.647886$ \\
    38e38o (76q) & 1500 & $-3442.030816$ & $-3441.902673$ & $-3441.907115$ & $-3441.733918$ \\
    44e44o (88q) & 1500 & $-3442.075085$ & $-3441.946636$ & $-3441.952030$ & $-3441.781578$ \\
    50e50o (100q) & 2000 & $-3442.169287^{\dagger}$ & $-3442.045607^{\dagger}$ &$-3442.013096^{*}$ & ---  \\
    \bottomrule
    \end{tabular}
    \caption{\textbf{Reference DMRG energies (in Ha) for different multiplicity states calculated with different bond dimensions ($\chi$) for various active spaces.} Active spaces are denoted as number of electrons, number of spatial orbitals, and number of qubits in parentheses. ($\dagger$): S$_0$ and S$_1$ were obtained using optimal orbital reordering~\cite{DMRG_CHEM_2}, allowing escape from local minima. (*): For this specific root, the lowest energy was obtained with standard ordering, leading to an unusual S$_0 <$ S$_1 <$ T$_1$ state ordering. For the 100q case, following the trend of previous active spaces for which T$_1$ is the lowest non-singlet root, we omitted the Q$_1$ calculation.}
    \label{tab:tld-1433-overlap-eigenvalues}
\end{table}

\section{Electronic Structure Problem}\label{sec:esp} 

In the second quantization representation, the Fock space of $N$ spin-orbitals (or Fermionic modes) is spanned by the set of basis states or Slater determinants, $|D\rangle = |n_1, n_2, \dots, n_N\rangle$, where $n_p \in\{0,1\}$ denotes the occupation number of the $p$-th spin-orbital. 
The Fermionic creation and annihilation operators that act on this space are defined as
\begin{align}
a_p^\dagger \, |n_1, n_2, \dots, 0_p, \dots, n_N\rangle 
= (-1)^{\sum_{q=1}^{p-1} n_q} \, |n_1, n_2, \dots, 1_p, \dots, n_N\rangle,
\end{align}
and 
\begin{align}
a_p \, |n_1, n_2, \dots, 1_p, \dots, n_N\rangle
= (-1)^{\sum_{q=1}^{p-1} n_q} \, |n_1, n_2, \dots, 0_p, \dots, n_N\rangle
\end{align}
respectively. 
These operators satisfy the anticommutation relations
 \begin{equation}
\{a_p, a_q\} = \{a_p^{\dagger}, a_q^{\dagger}\} = 0, 
\qquad 
\{a_p, a_q^{\dagger}\} = \delta_{pq}.
\end{equation}
The \textit{Fermionic vacuum state}
 $\ket{\mathrm{vac}_f}$ is a unique state satisfying
$a_p \ket{\mathrm{vac}_f} = 0$ for all $  p = 1, \ldots, N$.
All other basis states, collectively forming the Fock space, are generated by applying creation operators to the vacuum:
\begin{equation}
\ket{n_1 n_2 \ldots n_N}
    := \prod_{p=1}^{N} (a_p^\dagger)^{n_p} \ket{\mathrm{vac}_f}.
\label{eq:FockBasis}
\end{equation}

The $N$-mode second-quantized Fermionic Hamiltonian can be expressed in terms of creation and annihilation operators as
\begin{align}
H=\sum_{p, q=1}^{N} h_{pq} a_p^{\dagger} a_q+\sum_{p, q, r, s=1}^{N} h_{p q r s} a_p^{\dagger} a_q^{\dagger} a_r a_s,
\end{align}
where $h_{pq}$ and $h_{p q r s}$ are the one-electron and two-electron integrals, respectively. 

Alternatively, the Fermionic space can be defined using Majorana operators $\{m_k\}_{k=1}^{2N}$ where
\begin{equation}\label{eq:majorana}
    m_{2p-1}:= a^{\dagger}_p + a_p, \qquad m_{2p}:= \text{i}(a^{\dagger}_p - a_p).
\end{equation}
The $2N$ operators are algebraically independent, meaning that no two distinct products of Majorana operators, when ordered by index, yield the same operator. These operators are unitary, self-adjoint, and satisfy $\{m_p,m_q\} = 2 \delta_{pq}$. Products of Majorana operators $M_\nu = \text{i}^{r_\nu} m^{\nu_1}_1m^{\nu_2}_2 \cdots m^{\nu_{2N}}_{2N}$ can be represented uniquely up to a sign by a binary vector $\nu=(\nu_1,\dots,\nu_{2N})$ where $\nu_i \in \{0,1\}$, and $r_\nu = |\nu|(|\nu|-1)/2$, where $|\nu|$ is the Hamming weight of $\nu$, called its length.  The resulting $4^N$ Hermitian operators $\{M_\nu\}_\nu$, referred to as Majorana monomials, form a basis for Hermitian operators. 

Using Majorana monomials, the $N$-mode second-quantized Fermionic Hamiltonian can be expressed as
\begin{equation}
H=\sum_{p, q=1}^{2N} \text{i}c_{pq} m_im_j+\sum_{p, q, r, s=1}^{2N} c_{p q s r} m_p m_q m_r m_s, 
\end{equation}
for real coefficients $c_{p q}$ and $c_{p q s r}$. This is a direct consequence of the relationship given in Eq.~\eqref{eq:majorana}. The expression can be restated as
\begin{equation}\label{eq:mham}
H = \sum_\nu \alpha_\nu M_\nu,
\end{equation}
where $M_\nu$ are Majorana monomials of length 2 or 4.

\section{Majorana Propagation}\label{sec:majoranaprop}

\paragraph{Surrogate Graph}
One key aspect of the operator-length truncation criterion is that it enables the construction of a so-called surrogate graph.
This graph symbolically tracks the propagation of monomials through the circuit, allowing different parameters $\vec{\theta}$ to be plugged in for rapid evaluation without re-evolving the entire operator for each parameter choice.
This idea has been leveraged in prior work on Pauli propagation~\cite{rudolph2023classical}; here, we use it to substantially accelerate the optimisation of gate parameters.
The surrogate graph is possible because low-length monomials dominate the contribution to the expectation value, and the truncation criterion is independent of the gate parameters.
This contrasts with propagation schemes that rank terms by coefficient magnitude~\cite{ryabinkin2020iterative, steiger2024sparse, ryabinkin2025optimization}: in those cases, the surrogate graph would depend on the gate parameters and would need to be reconstructed for every parameter vector, resulting in a non-continuous landscape that complicates parameter optimization.
Our method, after precompilation of the surrogate graph, allows for fast evaluation of expectation values and provides analytically provable convergence guarantees~\cite{miller2025simulation}.

\paragraph{Truncation criteria}

A central ingredient in MP is the truncation criterion used during propagation to curb the exponential growth in the number of terms.
This choice directly impacts both efficiency and accuracy.
MP uses \textit{operator length} as its truncation rule, defined as the number of Majorana operators appearing in a monomial:
monomials whose length exceeds a chosen threshold are discarded during propagation.

The original analysis of MP~\cite{miller2025simulation} shows that for unstructured circuits composed of random length-4 monomial generators, high-length monomials are exponentially unlikely to affect expectation values.
The intuition is twofold. First, in a system with $N$ spin orbitals (or $2N$ Majorana modes), there are $\binom{2N}{\ell} \sim \mathcal{O}(N^{\ell})$ monomials of length $\ell$, which grows polynomially with $N$ for fixed $\ell$.
Thus, operator length induces a natural hierarchy of polynomial subspaces whose dimension increases exponentially with $\ell$.
Second, consider the overlap of some monomial with Fock-basis states.
Only \emph{paired} monomials have non-zero overlap, and for a given even length $2\ell$, there are only $\binom{N}{\ell} \sim \mathcal{O}(N^{\ell})$ paired monomials, a vanishing fraction of the $\binom{2N}{2\ell}$ monomials of length $2\ell$ as $N$ grows.
Consequently, high-length monomials rarely overlap with the Fock-basis state and contribute negligibly to the expectation value.

Beyond this exponential suppression, MP also exhibits a \textit{quadratic backflow} suppression from high-length to low-length monomials during propagation.
Even when high-length terms are generated, their influence on the low-length subspace (which dominates the expectation value) is quadratically damped.
Together, these effects justify operator length as a robust truncation criterion: discarding high-length monomials greatly reduces computational cost while maintaining a controlled approximation error.
The analysis in~\cite{miller2025simulation} provides rigorous guarantees under this scheme: the approximation error decays exponentially with the truncation threshold, and computing expectation values to a fixed error for typical circuit ensembles requires only polynomial resources.
These guarantees assume unstructured circuits composed of random length-4 monomial generators and random second-quantized Hamiltonians, assumptions closely aligned with ADAPT-VQE settings.

In addition to the length-based truncation criterion, one can also consider the following alternative schemes:
\begin{itemize}
    \item \textbf{Generalised length truncation}:
    Analogous to Pauli weight in qubit systems, where $\{X_i, Y_i, Z_i\}$ are anticommuting qubit-local operators each contributing one unit to the weight, we define a generalized length for Majorana monomials. On each mode $j$, the Fermionic-locally acting anticommuting operators are $\{m_{2j-1},\, m_{2j},\, m_{2j-1}m_{2j}\}$.
    The generalized length of a monomial is defined as the number of Fermionic modes it acts on, irrespective of whether the action is by a single Majorana operator or a paired operator.
    With this definition, the number of monomials of fixed generalized length $\ell$ is
    $\binom{N}{\ell} 3^\ell = \mathcal{O}(N^\ell).$
    Among these, there are $\binom{N}{\ell}$ fully paired monomials, which have standard Majorana length $2\ell$. At a fixed cutoff $\ell$, the generalized scheme therefore retains all paired monomials that would only be included at cutoff $2\ell$ under the standard Majorana-length definition. Equivalently, to include the same set of paired monomials, the standard scheme must raise its cutoff from $\ell$ to $2\ell$, increasing the total number of retained monomials from $\mathcal{O}(N^\ell)$ to $\mathcal{O}(N^{2\ell})$.
    This criterion is natural in the Heisenberg picture, where the objective is to parameterize operator growth paths from the initial observable to paired monomials that overlap with the state. 
    Under this redefinition, the same paths are retained at fixed scaling, and the theoretical analysis of Ref.~\cite{miller2025simulation} continues to apply.
    \item \textbf{Coefficient-magnitude truncation}:
          Discard monomials with $|c_\nu| < \tau$ for a chosen threshold $\tau$, thereby controlling graph growth during propagation and retaining only terms likely to contribute.
    \item \textbf{Coefficient-magnitude acceptance}:
          Retain monomials with $|c_\nu| \ge \tau$ irrespective of their length, prioritising potentially impactful contributions even when they are of high length.
    \item \textbf{Paired acceptance}:
        Retain all paired monomials, irrespective of their length.
        These are precisely the monomials that contribute to expectation values in the Heisenberg picture and are therefore always retained in our Heisenberg-picture simulations, in combination with length-based truncation.
        Although this acceptance rule permits exponentially many terms in principle, in practice it does not lead to uncontrolled growth. 
        When the circuit is generated by length-two and length-four Majorana monomial rotations, paired monomials above the length cutoff cannot transform into other paired monomials without passing through an intermediate high-length unpaired monomial. 
        Since such intermediate terms are not accepted, the propagation of paired monomials remains effectively bounded.
\end{itemize}

Applying these alternative schemes yields different efficiency–accuracy trade-offs.
Coefficient-magnitude truncation typically accelerates evaluation, but the induced truncation depends on the gate parameters, preventing a single surrogate graph from being reused across different $\vec{\theta}$.
Conversely, coefficient-magnitude acceptance can improve accuracy by keeping large contributions even at high length, at the cost of larger graphs.
When a good prior parameter $\vec{\theta}_0$ is available (e.g., from rotosolve~\cite{ostaszewski2021structure}), one can build a locally faithful surrogate.
Fix a neighbourhood $\Theta = \{\vec{\theta}:\ |\vec{\theta}-\vec{\theta}_0|\le \vec{\delta}\}$ and construct the graph by: (i) accepting monomials that exceed the acceptance threshold somewhere in $\Theta$, and (ii) truncating monomials that remain below the truncation threshold throughout $\Theta$.
This retains terms that may become significant under plausible parameter updates while discarding those unlikely to matter, enabling fast, reusable evaluations within $\Theta$. In contrast to methods from the literature~\cite{ryabinkin2020iterative, steiger2024sparse, ryabinkin2025optimization}, we do not use coefficient-magnitude truncation or acceptance to control the size of the Hamiltonian. 
In our experiments, (generalized) length truncation is the only criterion that ensures the state remains polynomial in size.

\paragraph{\label{app:schrodinger equivalence}Equivalence to Schr\"odinger picture simulation}
While MP is naturally formulated in the Heisenberg picture, it can be shown to be equivalent to a Schr\"odinger picture simulation under length truncation.
This equivalence is particularly useful when considering the addition of gates to an existing circuit, as in ADAPT-VQE, where one often prefers to add a gate at the end of the circuit rather than at the beginning.

To understand this equivalence, consider the evolution of the initial state $\varrho$ through the circuit $C_L(\vec{\theta}\,)$ in the Schr\"odinger picture.
Since $\varrho$ contains exponentially many terms, a truncation scheme is again required to control growth.
We begin by truncating the initial state $\varrho$ to include only paired monomials up to length $2\ell$.
We then propagate each monomial in $\varrho$ through the circuit $C_L(\vec{\theta}\,)$ using the same branching rules as in Equation~\eqref{eq:mp-dispatch}, but with a negative sign in front of the $\sin \theta_j$ term due to the reversed direction of propagation.
At each step, monomials whose length exceeds $2\ell$ are discarded.
After propagating through all gates, we compute the expectation value by evaluating the overlap with the observable $H$.

The equivalence between the Schr\"odinger and Heisenberg pictures under length truncation arises because both approaches effectively track the same computational paths.
In the Heisenberg picture, we track paths from the observable to the state; in the Schr\"odinger picture, we track paths from the initial state to the observable.
In both cases, the length-based truncation criterion ensures that only paths involving low-length monomials are retained.
Since only paired monomials of length up to $2\ell$ contribute to the final expectation value, both pictures consider the same set of computational paths, and consequently yield identical expectation values for any choice of parameters $\vec{\theta}$.
When coefficient-based truncation schemes are employed, however, this equivalence no longer holds, as the sets of retained paths differ between the two pictures depending on the parameter values.

\paragraph{Analytical gradient with respect to circuit parameters}
Variational quantum algorithms rely on classical optimizers to determine optimal circuit parameters.
For VQEs, gradient-based methods such as the Broyden-Fletcher-Goldfarb-Shanno (BFGS) algorithm are commonly employed.
These optimisation methods require computation of the gradient, which can be approximated using finite differences or computed exactly using parameter-shift methods for certain gate types~\cite{mitarai2018quantum}. Thanks to the surrogate graph, our approach enables analytical computation of the surrogate model's gradient.

To compute the approximate energy gradient, we express the gradient element with respect to a single parameter $\theta_j$ of the $j$-th gate as
\begin{equation}
    \frac{\partial}{\partial{\theta_j}} \tilde f(\vec\theta) = \Tr[\varrho_{j-1}\, \frac{\partial}{\partial{\theta_j}}\left(U_j^{\dagger} H_{j+1}\, U_j\right)],
\end{equation}
where $H_{j+1}$ is the observable back-propagated through gates $(j+1)$ to $L$, and $\varrho_{j-1}$ is the state forward-propagated through gates $1$ to $(j-1)$.
At the level of individual Majorana monomials, taking the parameter derivative yields the following branching rule for propagation through a single gate:
\begin{equation}
    \frac{\partial}{\partial{\theta_j}}U_j^{\dagger} M_\nu\, U_j \rightarrow
    \begin{cases}
        0                                                                           & \text{if } [M_\nu, M_{\gamma_j}] = 0,                             \\
        -\sin(\theta_j)\, M_\nu + \mathrm{i} \, \cos(\theta_j)\, M_{\gamma_j} M_\nu & \text{if } \{M_\nu, M_{\gamma_j}\} = 0, \\
        -\sin(\theta_j)\, M_\nu                                                     & \text{if } \{M_\nu, M_{\gamma_j}\} = 0 \text{ and truncated}.
    \end{cases}
\end{equation}
We implement this using a sliding-window scheme that maintains copies of both the evolved state and the Hamiltonian. The state is propagated forward up to gate $j-1$, while the observable is back-propagated from gate $j+1$ to gate $L$. 
The derivative with respect to the gate parameter is then evaluated using the expression above.
When multiple gates share the same parameter, the chain rule applies and we sum the individual contributions to obtain the total gradient.

This approach allows analytical computation of the surrogate model's gradient in time proportional to a single energy evaluation. 
Table~\ref{tab:mbs-benchmarks} reports the time spent on surrogate-graph preparation (preprocessing), energy evaluation, and joint energy-and-gradient evaluation. Preparing the surrogate graph is the dominant cost but is performed only once per parameter-optimization step. In contrast, energy and gradient evaluations are substantially faster, and computing the gradient in addition to the energy incurs at most a threefold overhead relative to energy evaluation alone, independent of system size and the number of ADAPT-VMPE iterations.

\begin{table}[]
    \centering
    \begin{tabular}{@{}cc@{\qquad}ccc@{}}
    \toprule
         \#qubits & \# monomials  & \makecell{preprocessing\\time[s]}  & \makecell{energy estimation\\time [s]} & \makecell{energy and gradient\\estimation time[s]} \\\midrule
28  & 100       & 1.0072        & 0.0136        & 0.0358 \\
28  & 300       & 3.4245        & 0.0364        & 0.1034 \\
40  & 100       & 3.7154        & 0.0340        & 0.0842 \\
40  & 300       & 11.4904       & 0.0806        & 0.1810 \\
52  & 100       & 11.5886       & 0.1401        & 0.2578 \\
52  & 300       & 30.6624       & 0.1399        & 0.4111 \\
64  & 100       & 30.2919       & 0.2346        & 0.3903 \\
64  & 300       & 68.7662       & 0.3279        & 0.6531 \\
76  & 100       & 80.6944       & 0.4539        & 0.6514 \\
76  & 300       & 150.8662      & 0.5877        & 1.2993 \\
\bottomrule
\end{tabular}
    \caption{\textbf{Time required for preprocessing, single energy computation, and single energy-and-gradient computation for various active space sizes and numbers of ADAPT gates (monomials).} We used the circuits produced by ADAPT-VMPE for Complex A. Times presented are averages over 60 runs.}
    \label{tab:mbs-benchmarks}
\end{table}

\paragraph{Efficient gradient selection for ADAPT-VQE}

The ADAPT-VQE algorithm iteratively constructs a parameterized quantum circuit by selecting gates from an operator pool that best improve the variational energy.
At iteration $L$, the algorithm selects the unitary
\[
U_L = e^{-i \theta_L M_{\gamma_L}/2},
\]
which maximizes the magnitude of the energy gradient with respect to $\theta_L$ evaluated at $\theta_L = 0$.

In the Heisenberg picture, the selection gradient is given by
\begin{equation}
\left.
\frac{\partial}{\partial \theta_L}
\Tr\!\left[ U_L \, \varrho \, U_L^{\dagger} H_{L-1} \right]
\right|_{\theta_L = 0}
=
- \Tr\!\left[ H_{L-1} \, [ M_{\gamma_L}, \varrho ] \right],
\end{equation}
while in the Schr\"odinger picture it can be written as
\begin{equation}
\left.
\frac{\partial}{\partial \theta_L}
\Tr\!\left[ \varrho_{L-1} \, U_L^{\dagger} H U_L \right]
\right|_{\theta_L = 0}
=
\Tr\!\left[ \varrho_{L-1} \, [ M_{\gamma_L}, H ] \right].
\end{equation}
In these expressions, the newly selected gate is applied to the initial state or Hamiltonian, respectively.

To evaluate the commutator $[ M_{\gamma_L}, \varrho ]$ in the Heisenberg picture, it is sufficient to retain only those terms that contribute to the trace. Assuming that only length-two and length-four Majorana operators are in the pool, and noting that a length-four Majorana generator can increase or decrease the length of a Majorana monomial in the state by at most two, capturing all terms that may appear in the evolved Hamiltonian requires setting the maximum monomial length in the initial state to $\ell + 2$, where $\ell$ is the cutoff used for the Hamiltonian. The number of terms required in the initial state, therefore, needs to be $\mathcal{O}\!\left(N^{(\ell+2)/2}\right)$ to include all monomials relevant to the gradient computation.

One can then compute the commutator $[ M_{\gamma_L}, \varrho ]$ for each gate in the pool and search for it in the evolved Hamiltonian representation with $\mathcal{O}(1)$ complexity, since the evolved operator is stored as a hash map. Thus, for cutoff-$6$, the complexity is $\mathcal{O}(K N^4)$, where $K$ is the number of elements in the pool.

In the Schr\"odinger case, we consider the commutators $[ M_{\gamma_L}, H ]$. In general, there are $\order{N^4}$ terms in second quantized Hamiltonians, so computing the commutator for each gate has complexity $\mathcal{O}(K N^4)$ for $K$ pool elements, which is independent of the cutoff.

\section{Performance Improvements to ADAPT-VMPE} \label{sec:adapt-vmpe-performance}
Gate selection is a vital and time-consuming part of ADAPT-VMPE. This follows from the fact that a potentially large description of the evolved Hamiltonian (up to $\order{N^{\rm cutoff}}$ terms) must be evaluated against $\Theta(N^4)$ pool candidates. While polynomial in principle, the time required for pool evaluation grows rapidly with the number of qubits. Below we present several non-exclusive methods by which this process can be substantially accelerated with negligible impact on circuit quality.

\subsection{Pool reduction for Majoranic pool for Heisenberg ADAPT-VMPE}
Suppose $M_\nu$ and $M_{\nu'}$ are even-length Majorana monomials acting on the same set of modes, with the same parity of odd indices, at each mode acting with only one Majorana operator.
Then there exists an even-size set $J$ such that
\begin{equation}
    M_{\nu'} = \eta\, M_\nu Q, \qquad Q := \prod_{j\in J} \bar m_j, \qquad \eta\in\{\pm 1\},
\end{equation}
where each $\bar m_j=-i m_{2j-1}m_{2j}$ is a paired Majorana operator.
Because $|J|$ is even, $Q^2=I$. Also, $Q$ is diagonal in the Fock basis, so for any Fock state $\ket{\vec n}$,
\begin{equation}
    Q\ket{\vec n}=s_J\ket{\vec n},\qquad s_J\in\{\pm 1\}.
\end{equation}
Using $e^{i\theta A}=\cos(\theta)I+i\sin(\theta)A$ for involutory Hermitian $A$,
\begin{equation}
\begin{split}
    e^{i\theta M_{\nu'}}\ket{\vec n}
    &=\left(\cos\theta+i\sin\theta\,\eta M_\nu Q\right)\ket{\vec n}
    =\left(\cos\theta+i\eta s_J\sin\theta\,M_\nu\right)\ket{\vec n} \\
    &= \left(\cos\theta+i\sin(\eta s_J\theta)\,M_\nu\right)\ket{\vec n} = \left(\cos\theta+i\sin(\pm\theta)\,M_\nu\right)\ket{\vec n} \\
    &= \left(\cos(\pm\theta)+i\sin(\pm\theta)\,M_\nu\right)\ket{\vec n} = e^{i\pm\theta M_{\nu}}\ket{\vec n},
\end{split}
\end{equation}
which shows that the effect of $M_{\nu'} $ and $M_{\nu}$ on a Fock state is the same up to the change of the phase of $\theta$.
This is the same orbit-equivalence mechanism observed in the Jordan-Wigner mapping in qubit space~\cite{ryabinkin2020iterative}.

\subsection{Pool trimming}

In the ADAPT-VMPE procedure, all gates in the predefined pool $P$ must be evaluated at each step to select the gate to be added to the ansatz. However, the size of typical pools consisting of single and double excitations grows as $\order{N^4}$ with the number of qubits $N$, making the pool evaluation step time-consuming for large systems. One possible solution is to evaluate only a subset $P_i \subseteq P$ of the pool at each step $i$. To our knowledge, few studies address this problem, except~\cite{zhang2021mutual}, which uses mutual information from classical methods like DMRG to reduce the pool, and the recent work~\cite{larose2026cost}, in which the best pool elements are retained and supplemented with random operators each iteration.

Fixed pool trimming relies on the observation that pool elements that fail to improve the current energy at a given point are unlikely to improve it in later iterations either. This allows some elements to be removed from the pool, reducing the number of candidates to inspect. This observation is illustrated in Fig.~\ref{fig:ranks}, which analyzes how the selected pool elements at each iteration performed across subsequent iterations. The ranks determined from the first iteration show that the same ansatz would have been obtained using only the best 8,000 elements. In a more extreme case, keeping only the best 250 pool elements for the first 20 iterations yields the same ansatz. As the number of iterations increases, the ranks of the selected elements (as determined from the first iteration) grow; recomputing ranks at iteration 26 again yields smaller ranks for a certain number of subsequent iterations, with the same pattern repeating at iterations 51 and 76.

\begin{figure}[htbp]
    \centering
    \includegraphics[width=0.6\textwidth]{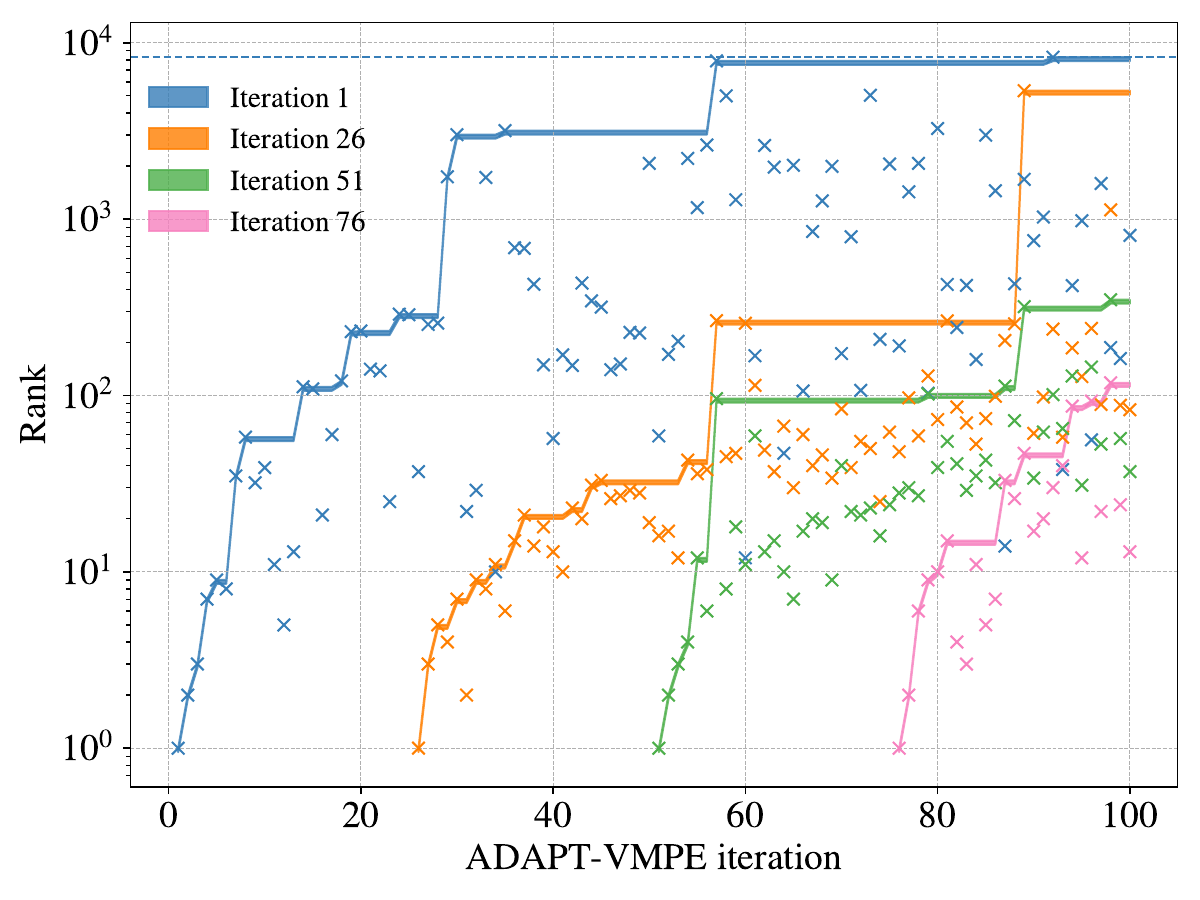}
    \caption{\textbf{Analysis of rank persistence among ADAPT-VMPE iterations for various reference ranks.} Ranks of the selected elements are computed at iterations 1, 26, 51, and 76 (rank 1 indicates the best gate). Marker placed at $(x,y)$ with color corresponding to $k\in \{1, 26, 51, 76\}$ iteration, should be understood as gate selected at $x$-th iteration was assigned a rank $y$ at $k$-th iteration.. The simulation is performed using ADAPT-VMPE with the L-BFGS-B optimizer in the Heisenberg setting with the Majoranic pool for the 40-qubit Complex A Hamiltonian. The pool contains 14,250 elements.}
    \label{fig:ranks}
\end{figure}

To implement the pool trimming procedure, we start with $P_1=P$. Each pool element $p_j \in P$ is assigned a rank $r^{1}_j\in \{1,\dots, |P|\}$ based on the evaluation results. Pool elements whose ranks exceed a threshold $\tau$ are removed, giving $\tilde P=\{p_j~|~p\in P \mbox{ and } r^{1}_j < \tau \}$. For a predetermined number of steps $k=2, 3,\dots,\kappa-1$, ADAPT-VMPE uses the trimmed pool $P_k=\tilde P$, after which the full pool is re-evaluated by setting $P_\kappa = P$, yielding a new set of ranks $\{r^{\kappa}_j\}$. This full re-evaluation is motivated by Fig.~\ref{fig:ranks}, which shows that as the iteration number grows, the ranks of selected elements also increase, indicating the need for periodic re-ranking. 

There are two hyperparameters: pool cutoff threshold $\tau$ (the number of pool elements to retain) and $\kappa$ (the frequency of full pool re-evaluation). For example, for the 40-qubit Complex A case, setting $\tau=300$ and $\kappa=25$ yields the same ansatz after 100 iterations while providing a significant reduction in the number of gate evaluations.

\subsection{Mixed selection procedure}

Two methods dominate the literature for gate selection. The first, gradient selection, prefers gates $U(\theta)$ that maximize the magnitude of the derivative of $\Tr(\varrho U^\dagger (\theta)H U(\theta))$ at $\theta=0$: the larger this value, the greater the expected improvement when the corresponding gate is added. The second, GGF selection~\cite{feniou2025greedy}, directly estimates the expected energy improvement when optimizing $\min_\theta \Tr(\varrho U^\dagger (\theta)H U(\theta))$. Depending on the generator $G$, finding the optimal $\theta$ may require evaluating the expectation value at only a few values of $\theta$. In particular, if $G$ is a Pauli string or Majorana monomial, evaluations at only 3 values are needed. A particular benefit of this method is that it not only selects a gate to be appended to the ansatz, but also proposes a good initial value for parameter optimization.

While both methods can improve the current quantum state, we found that GGF selection outperforms gradient selection, as shown in Fig.~\ref{fig:rotograd}. On the other hand, our implementation allows much faster gradient evaluation than GGF selection. This motivates using gradient selection for pool trimming (as a fast approximate assessment) and GGF selection for the intermediate gate selection steps from the trimmed pool. As shown in Fig.~\ref{fig:rotograd}, this mixed approach matches GGF selection in quality while being substantially faster. 

\begin{figure}[htbp]
    \centering
    \includegraphics[]{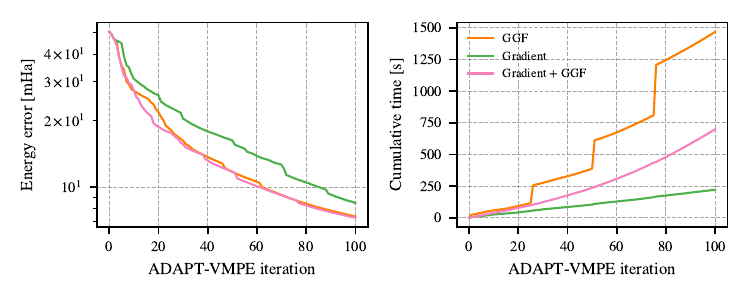}
    \caption{\textbf{Comparison of convergence and performance of GGF, gradient-based, and mixed gate selection.} The simulation is performed using ADAPT-VMPE with the L-BFGS-B optimizer in the Heisenberg setting with the Majoranic pool for the 40-qubit Complex A Hamiltonian, using fixed trimming with $\tau = 285$ and $\kappa = 25$. ``Gradient + GGF'' identifies the case where gradient selection is used only when the full pool is evaluated, and GGF selection is used in intermediate steps. The left plot shows that gradient selection alone performs worse than GGF selection and the mixed approach. Despite its quality, GGF selection alone is time-consuming, as shown in the right plot. Combining gradient and GGF, using the former for pool trimming and the latter for selection, substantially speeds up gate selection.}
    \label{fig:rotograd}
\end{figure}

\section{Transpilation cost analysis} \label{appendix:transpilation-cost-analysis}

\subsection{All-to-all connectivity} For all-to-all connectivity, implementing Pauli exponentiation $\exp(\mathrm i \theta P)$ for a $k$-local operator requires $2k-2=\order{k}$ CNOT operators, which is optimal. Both iQCC and vDBF, due to pool synthesis, can in principle select an arbitrary operator, which is why their CNOT cost may reach $\order{N}$ in the worst case. In practice, however, both algorithms are expected to select much cheaper operators. iQCC typically selects operators that are at most 4-local; these reflect Fermionic double excitation operators up to the missing $Z$ Pauli strings removed by the authors, as confirmed by the log files in~\cite{github-iqcc} for N$_2$, giving $\order{1}$ cost per gate added. For similar reasons, vDBF is expected to typically select 4-length Majorana monomials. With the ternary-tree F2Q mapping~\cite{jiang2020optimal}, each Majorana operator becomes $\order{\log N}$-local, so the product of 4 such operators is also $\order{\log N}$-local, giving a typical CNOT cost of $\order{\log N}$ for vDBF.

SWCS, when using Fermionic-spin pool, can attain benefits from proper Fermionic-spin selection, yielding a guaranteed $\order{\log N}$ CNOT cost per gate. Otherwise, when using CEO pool, a $\order{1}$ CNOT cost per gates is achieved.

ADAPT-VMPE uses only 2- and 4-length Majorana monomials by construction. As for vDBF, the typical and worst-case CNOT cost is $\order{\log N}$.

\subsection{Heavy-hexagonal connectivity} For limited connectivity, implementing $\exp(\mathrm i \theta P)$ requires $\order{k_{\rm ST}}$ CNOTs, where $k_{\rm ST}$ is the number of nodes in the Steiner tree spanned by the qubits on which $P$ acts. This complexity can be achieved either via the Steiner tree implementation of~\cite{miller2024treespilation} or via an appropriate SWAP sequence. The value $k_{\rm ST}$ may be as small as $k$, but can also be much larger; one can formally show that $k_{\rm ST} = \order{dk}$ where $d$ is the maximum inter-node distance in the connectivity graph, which is $\order{\sqrt N}$ for heavy-hexagonal or 2D grid topologies. This follows by noting that a valid Steiner tree\footnote{We use `Steiner tree' to mean any tree connecting all the terminals, not necessarily the optimal one.} can be constructed by successively connecting vertices via shortest paths and removing any resulting cycles. For randomly selected qubits, as is typical in ADAPT-VQE-like ansatz construction, this scaling is expected to hold.

From these results, iQCC requires $\order{\sqrt N}$ CNOTs in the typical scenario (and $\order{N}$ in the worst case), using the arguments from the previous subsection. Similarly, each gate added by SWCS requires $\order{\sqrt N}$ CNOTs when CEO pool is used. 

For vDBF, ADAPT-VMPE and SWCS using Fermionic-spin pool, using the optimal ternary-tree mapping produces $\order{\log N}$-local Pauli operators, requiring $\order{\sqrt{N}\log N}$ CNOTs. However, using a Bonsai F2Q mapping, it can be formally shown that only $\order{\sqrt N}$ CNOTs are required~\cite{miller2023bonsai}, and the optimal Steiner tree for such mappings can be found in polynomial time~\cite{miller2024treespilation}.

\subsection{Number of RZ gates}
Since iQCC, vDBF, and ADAPT-VMPE each add a single-Pauli-exponentiation operator per iteration, the number of RZ gates per gate added equals 1. For SWCS, either the Fermionic-spin pool is used~\cite{mullinax2025large}, which requires 8 RZ gates per operator upon decomposition, or the CEO pool~\cite{ramoa2025reducing}, which requires 4.

\section{Compilation methods for large scale benchmarks} \label{appendix:compilation}

\subsection{NISQ} For all-to-all connectivity, we used the standard treespilation algorithm with Pauli string cost as described in~\cite{miller2024treespilation}. We ran simulated annealing with a geometric schedule $T\cdot 0.9995^t$, starting from 50 and ending at 1. Once the F2Q mapping is selected, we transform the Fermionic circuit into a qubit circuit, transpile using the Steiner method from~\cite{miller2024treespilation}, and further fine-tune with the Qiskit transpiler at optimization level 2, with basis gates \texttt{cx} and \texttt{u}. Active rotations were not transpiled, as they are applied to the Hamiltonian.

For heavy-hexagonal connectivity, we used the topology of the IBM Fez quantum computer and ran treespilation to find a Bonsai F2Q mapping whose ternary tree preserves the topology. This also provides an initial layout used for the Steiner method. We then applied Qiskit transpilation with the IBM Fez topology using the same settings as for all-to-all connectivity. 

\subsection{Fault-tolerant} For FT algorithms like QPE, the cost also depends on the number of Paulis in the operator and the $L_1$ norm of the operator coefficients. Since active rotations may increase this number, it might in general be preferable to implement them as part of the state preparation circuit. However, for the molecules considered in this paper, this does not appear to be the case, as presented in Fig.~\ref{fig:l1-norm}: operators dressed with active rotations appear to be simpler than the originals. Therefore, in the main text we present the T count for the case where active rotations are dressed to the Hamiltonian, while Fig.~\ref{fig:compilation-ft} presents both variants.

\begin{figure}
    \centering
    \includegraphics[width=\linewidth]{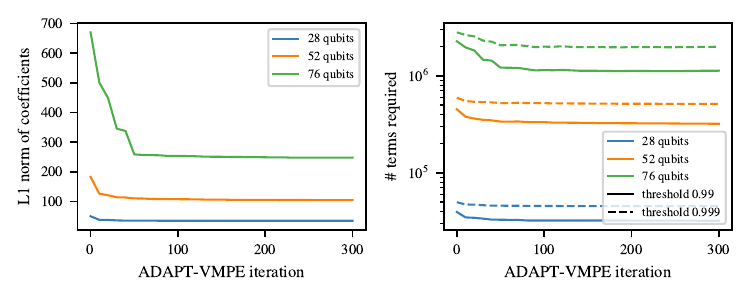}
    \caption{\textbf{Change of the $L_1$ norm of operator coefficients and the number of significant terms with consecutive actively rotated Hamiltonians.} The number of terms is computed as the number required to reach $p\|H\|$ norm for $p=99\%$ and $p=99.9\%$, where $H$ is the original Hamiltonian. Both measures are commonly used to quantify the complexity of the Hamiltonian when used in QPE. Active rotations, as optimized during the ADAPT-VMPE iterations, tend to simplify the Hamiltonian: the $L_1$ norm of the coefficients decreases substantially (left), and the number of Pauli operators required to reach 99\% and 99.9\% of the operator norm also decreases with iterations. Dressing active rotations is an exact unitary transformation of the Hamiltonian, so the spectrum is preserved.}
    \label{fig:l1-norm}
\end{figure}
\begin{figure}
    \centering
    \includegraphics[width=\linewidth]{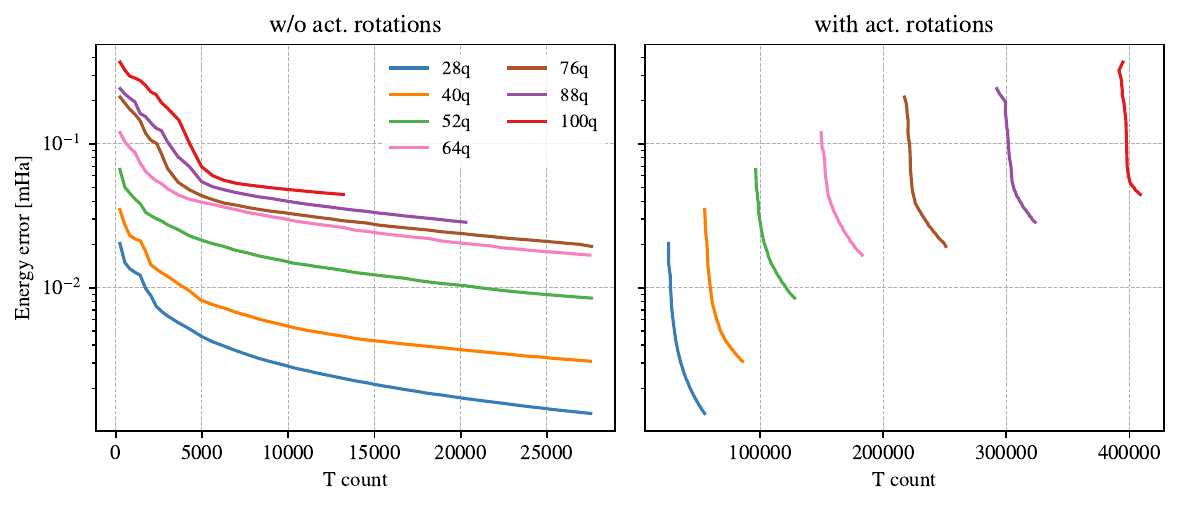}
    \caption{\textbf{Energy error as a function of T count cost for (a) without and (b) with active rotations included in transpilation.} In (a), the active rotations are dressed to the Hamiltonian; in (b), they are transpiled as part of the ansatz.}
    \label{fig:compilation-ft}
\end{figure}

Since limited connectivity and the choice of Bonsai F2Q mapping have no impact on the number of T gates (as SWAP operators are Clifford), we selected Jordan-Wigner mapping and all-to-all connectivity. We then transpiled the qubit circuit into Clifford and $Z$-rotation gates using Steiner decomposition~\cite{miller2024treespilation}, and compiled further to obtain a circuit reproducing the original state with error at most $\varepsilon=10^{-4}$. Each $Z$-rotation was decomposed into single-qubit Clifford and T gates with infidelity $\frac{\varepsilon}{\#\mathrm{RZ}}$, using \texttt{pygridsynth}~\cite{pygridsynth}, an implementation of a variant of the Ross-Selinger algorithm~\cite{ross2014optimal}. After this step, the circuit consists only of Clifford and T gates, and we computed the count of the latter.

\section{Overlap bounds} \label{appendix:overlaps}

This section is dedicated to the formal derivation of the overlap of a quantum state $\ket{\psi}$ with the ground state of a molecular Hamiltonian $H$.

\begin{theorem}\label{theorem:simple-overlap}
Let $H$ be a molecular Hamiltonian and let $\ket{E_0}, \ket{E_1}$ be the ground state and first excited state with corresponding energies $E_0, E_1$. Let $E_0 < E_1$. Let $\ket{\psi}$ be a quantum state such that $E\coloneqq \bra \psi H \ket{\psi} \in (E_0, E_1)$. Then
\begin{equation}
    |\braket{\psi}{E_0}|^2 \geq 1- \frac{E-E_0}{E_1-E_0}.
\end{equation}
\end{theorem}
\begin{proof}
Let $\ket{\psi}$ be of the form $\ket{\psi} = \sum_i \alpha_i \ket{E_i}$, where $\ket{E_i}$ are eigenvectors of $H$. One has
 \begin{equation}
     \begin{split}
E- E_0 & = \bra \psi H \ket \psi - E_0 = \sum_{k,m} \alpha_k^\ast \alpha_m E_m \bra {E_k} \ket {E_m} - E_0 = \sum_k |\alpha _k |^2E_k - E_0 \\
& \geq (1-|\alpha_0|^2) E_1 + |\alpha_0|^2 E_0 - E_0  = (E_1-E_0) (1 - |\alpha_0|^2),
\end{split}
 \end{equation}
 and from this
 \begin{equation}
     1 - |\alpha_0|^2 \leq \frac{E-E_0}{E_1-E_0},
 \end{equation}
 which after basic transformation and given $|\alpha_0| = | \braket{\psi}{E_0}| $ proves the statement.
\end{proof}

The above theorem is useful when the only information about the state $\ket{\psi}$ is its energy. It effectively reflects the intuition that the worst-case overlap occurs when $\ket{\psi}$ is a superposition of $\ket{E_0}$ and $\ket{E_1}$. However, solutions to molecular Hamiltonians typically belong to a subspace defined by symmetry constraints on the number, $z$-spin, and $S^2$ operators. We assume hereafter that the ground state is a singlet from a particular particle sector; the theorems can be readily adapted to other cases. We will refer to a state as non-singlet if its spin multiplicity differs from singlet or if it belongs to the wrong particle-number sector.

\begin{theorem}\label{theorem:overlap-known-alpha}
Let $H$ be a Hamiltonian and let $\ket{E_0}, \ket{S^\top_1}, \ket{S_1}$ be the ground state (a singlet state), the lowest-energy non-singlet eigenstate, and the first singlet excited state, with corresponding energies $E_0, S^\top_1, S_1$. Let $E_0 < S_1$. Let $\ket{\psi}$ be a quantum state with expectation value $E \in (E_0, S_1)$. Let $\Pi_S$ be the projection onto the space of singlet states and let $|\alpha_{S^\top}|^2 = 1- \bra \psi \Pi_S \ket{\psi}$ be the overlap of the state with the non-singlet space. Then
\begin{equation}
    |\braket{\psi}{E_0}|^2 \geq  \frac{S_1-E}{S_1-E_0} - |\alpha_{S^\top}|^2 \frac{S_1-S^\top_1}{S_1-E_0}.
\end{equation}
\end{theorem}
\begin{proof}
Any quantum state can be written uniquely (up to a phase of the $\alpha$'s) in the following form
\begin{equation}
    \ket{\psi} = \alpha_0 \ket{E_0} + \alpha_S \ket{\psi_S} + \alpha_{S^\top} \ket{\psi_{S^\top}},
\end{equation}
where $\ket{\psi_S}$ is a state orthogonal to $\ket{E_0}$ satisfying $\bra{\psi_S} \Pi_S \ket{\psi_S} = 1$, and $\ket{\psi_{S^\top}}$ is the remaining part satisfying $\bra{\psi_{S^\top}} \Pi_S \ket{\psi_{S^\top}} = 0$. Assuming all states on the right-hand side are normalized,
\begin{equation}
    |\alpha_0|^2 + |\alpha_S|^2 + |\alpha_{S^\top}|^2 = 1. \label{eq:normalization}
\end{equation}
In addition,
\begin{equation}
\begin{split}
    E &= \bra \psi H \ket \psi \\
    &= (\alpha_0^\ast \bra{E_0} + \alpha_S^\ast \bra{\psi_S} + \alpha_{S^\top}^\ast \bra{\psi_{S^\top}}) H (\alpha_0 \ket{E_0} + \alpha_S \ket{\psi_S} + \alpha_{S^\top} \ket{\psi_{S^\top}}) \\
&    \geq |\alpha_0|^2 E_0 + |\alpha_S|^2 S_1 + |\alpha_{S^\top}|^2 S_1^\top  \\
&\phantom {=\ }+2\Re (\alpha_0^\ast  \alpha_{S} \bra{E_0} H \ket{\psi_{S}} )\\
&\phantom {=\ }+2\Re (\alpha_0^\ast  \alpha_{S^\top} \bra{E_0} H \ket{\psi_{S^\top}} )\\
&\phantom {=\ }+2\Re (\alpha_S^\ast  \alpha_{S^\top} \bra{\psi_S} H \ket{\psi_{S^\top}} )
\end{split}
\end{equation}
where we used $\bra{E_0}H\ket{E_0} =E_0$, $\bra{\psi_S}H\ket{\psi_S} \geq S_1$, $\bra{\psi_{S^\top}}H\ket{\psi_{S^\top}} \geq S^\top_1$. 
The last three terms all vanish. Indeed, 
\begin{equation}
    \bra{E_0} H \ket{\psi_{S}} = E_0 \bra{E_0} \ket{\psi_{S}} = 0
\end{equation}
by definition of $\ket{\psi_S}$; similarly $\bra{E_0} H \ket{\psi_{S^\top}}=0$. Finally, $\bra{\psi_S} H \ket{\psi_{S^\top}}=0$ because $[H, \Pi_S]=0$ (since $H$ is a molecular Hamiltonian), meaning singlet states are mapped to singlet states, and the singlet and non-singlet subspaces are orthogonal. This yields
\begin{equation}
\begin{split}
E &\geq |\alpha_0|^2 E_0 + |\alpha_S|^2 S_1 + |\alpha_{S^\top}|^2 S^\top_1 \\
& =|\alpha_0|^2 E_0 + (1- |\alpha_0|^2 - |\alpha_{S^\top}|^2) S_1 + |\alpha_{S^\top}|^2 S^\top_1 \\
& =S_1 + |\alpha_0|^2 (E_0 - S_1)  + |\alpha_{S^\top}|^2 ( S^\top_1 - S_1), 
\end{split}
\end{equation}
where we used Eq.~\eqref{eq:normalization} to eliminate $|\alpha_S|^2$. A simple algebraic rearrangement gives the required statement, noting that $E_0<S_1$ by assumption. 
\end{proof}
Theorem~\ref{theorem:overlap-known-alpha} provides a bound for both $S^\top_1 < S_1$ and $S^\top_1 > S_1$, improving upon Theorem~\ref{theorem:simple-overlap} in both cases. Indeed, if $S_1 < S^\top_1$, then $E_1 \equiv S_1$, making the first term of the bound equal to that of Theorem~\ref{theorem:simple-overlap}. The subtracted term then becomes negative, contributing positively to the bound for any $|\alpha|$. If $S^\top_1 < S_1$, assuming $E\in(E_0, S^\top_1)$ (the range required for Theorem~\ref{theorem:simple-overlap} with $E_1 \equiv S^\top_1$), the difference between the two bounds is
\begin{multline}
    \frac{S_1-E}{S_1-E_0} - |\alpha_{S^\top}|^2 \frac{S_1-S^\top_1}{S_1-E_0}  - \frac{S^\top_1-E}{S^\top_1-E_0}  \\= \frac{(S_1-E)(S^\top_1 - E_0) - |\alpha_{S^\top}|^2 (S_1-S^\top_1) (S^\top_1 - E_0) - (S^\top_1-E)(S_1-E_0)}{(S_1-E_0)(S^\top_1 - E_0)} \\
    = \frac{-S_1E_0- ES^\top_1+S_1^\top E_0 + S_1E - |\alpha_{S^\top}|^2 (S_1-S^\top_1) (S^\top_1 - E_0) }{(S_1-E_0)(S^\top_1 - E_0)} \\
    = \frac{S_1-S_1^\top}{(S_1-E_0)(S^\top_1 - E_0)}\big (E-E_0  - |\alpha_{S^\top}|^2 (S_1^\top - E_0)\big ).
\end{multline}
The leading fraction is positive by assumption. Let $|\beta|^2$ denote the overlap with the singlet space; since it is at least as large as the overlap with the ground state alone, it satisfies the same bound from Theorem~\ref{theorem:simple-overlap}. In addition, $|\beta|^2 = 1- |\alpha_{S^\top}|^2$. Therefore,
\begin{equation}
\begin{split}
    E-E_0  - |\alpha_{S^\top}|^2 (S_1^\top - E_0) & =  E-E_0  - (1-|\beta|^2) (S_1^\top - E_0) \\
    &= E - S_1^\top + |\beta|^2 (S_1^\top - E_0) \\
    & \geq  E - S_1^\top + \frac{S_1^\top - E}{S_1^\top - E_0} (S_1^\top - E_0) = 0,
\end{split}
\end{equation}
which confirms that the bound from Theorem~\ref{theorem:overlap-known-alpha} is tighter than that from Theorem~\ref{theorem:simple-overlap}.

Theorem~\ref{theorem:overlap-known-alpha} requires knowledge of the overlap $|\alpha_{S^\top}|^2$, which may in general be difficult to estimate. What is readily accessible is the expectation value with respect to number, $z$-spin and $S^2$ symmetries, or the corresponding penalty Hamiltonians. 
We now introduce a suitable penalty Hamiltonian.
\begin{definition}
Let $H_P$ be a Hermitian operator. We call it a $(\lambda_2,\lambda_P)$-\textit{singlet penalty Hamiltonian} if it satisfies all of the following conditions simultaneously:
\begin{enumerate}
    \item $H_P \geq 0$,
    \item the nullspace of $H_P$ is the space of all singlet states,
    \item the smallest positive eigenvalue of $H_P$ is bounded below by $\lambda_2$,
    \item the spectrum is bounded above by $\lambda_P$.
\end{enumerate}
\end{definition}

Such $H_P$ can be easily constructed by an appropriate combination of penalty Hamiltonians for the number, $z$-spin, and $S^2$ symmetries. The following theorem provides a bound when $|\alpha_{S^\top}|^2$ is unknown, but the expectation value of $\ket{\psi}$ with respect to $H_P$ is available. This value can be approximated, e.g., with the Majorana Propagation simulator or directly with a QPU.

\begin{theorem}\label{theorem:overlap-unknown-alpha}
Let $H, \ket{E_0}, \ket{S^\top_1}, \ket{S_1}, E_0, S^\top_1, S_1$ be as in Theorem~\ref{theorem:overlap-known-alpha}. Let $H_P$ be a $(\lambda_2,\lambda_P)$-singlet penalty Hamiltonian and let $p \coloneqq \bra{\psi}H_P\ket{\psi} \in [0, \lambda_2]$. Then
\begin{equation}
    |\braket{\psi}{E_0}|^2 \geq 
    \begin{cases}
     \frac{S_1-E}{S_1-E_0} - \frac{p}{\lambda_2}\frac{S_1-S^\top_1}{S_1-E_0}, &S^\top_1 <S_1,\\
     \frac{S_1-E}{S_1-E_0} - \frac{p}{\lambda_P} \frac{S_1-S^\top_1}{S_1-E_0}, & \rm otherwise.
    \end{cases}
\end{equation}
\end{theorem}
\begin{proof}
Observe that
\begin{equation}
\begin{split}
p & =\bra{\psi} H_P \ket{\psi} \\
&=(\alpha_0^\ast \bra{E_0} + \alpha_S^\ast \bra{\psi_S} +\alpha_{S^\top}^\ast \bra{\psi_{S^\top}}) H_P (\alpha_0 \ket{E_0} + \alpha_S \ket{\psi_S} + \alpha_{S^\top} \ket{\psi_{S^\top}}) \\
    &= |\alpha_{S^\top}|^2 \bra{\psi_{S^\top}} H_P \ket{\psi_{S^\top}}, 
\end{split}
\end{equation}
where we used $H_P\ket{E_0}=H_P\ket{\psi_S}=0$, since both states are singlet states. Since $\ket{\psi_{S^\top}}$ is a non-singlet state, $\bra{\psi_{S^\top}} H_P \ket{\psi_{S^\top}}\in [\lambda_2, \lambda_P]$ by definition of $H_P$. Therefore,
\begin{equation}
    |\alpha_{S^\top}|^2 = \frac{p}{\bra{\psi_{S^\top}} H_P \ket{\psi_{S^\top}}}
\end{equation}
gives
\begin{equation}
    \frac{p}{\lambda_P} \leq |\alpha_{S^\top}|^2 \leq \frac{p}{\lambda_2}.
\end{equation}

Since the exact value of $|\alpha_{S^\top}|^2$ is unknown, we must account for the worst possible case. This leads to different expressions depending on whether $S^\top_1<S_1$ or not. Assuming $S^\top_1 <S_1$ and using Theorem~\ref{theorem:overlap-known-alpha}:
\begin{equation}
    |\braket{\psi}{E_0}|^2 \geq  \frac{S_1-E}{S_1-E_0} - |\alpha_{S^\top}|^2 \frac{S_1-S^\top_1}{S_1-E_0} \geq \frac{S_1-E}{S_1-E_0} - \frac{p}{\lambda_2}\frac{S_1-S^\top_1}{S_1-E_0}.
\end{equation}
In the opposite case:
\begin{equation}
    |\braket{\psi}{E_0}|^2 \geq  \frac{S_1-E}{S_1-E_0} - |\alpha_{S^\top}|^2 \frac{S_1-S^\top_1}{S_1-E_0} \geq \frac{S_1-E}{S_1-E_0} - \frac{p}{\lambda_P} \frac{S_1-S^\top_1}{S_1-E_0}.
\end{equation}
These two cases give the required statement. 
\end{proof}

Both bounds in Theorem~\ref{theorem:overlap-unknown-alpha} coincide at $S^\top_1=S_1=E_1$, trivially recovering Theorem~\ref{theorem:simple-overlap}. We conclude with a generalization to the case where $S_1^\top$ is unknown.
\begin{theorem}
    Let $H, \ket{E_0}, \ket{S^\top_1}, \ket{S_1}, E_0, S_1, H_P, p$ be as in Theorem~\ref{theorem:overlap-unknown-alpha}. Then

    \begin{equation}
    |\braket{\psi}{E_0}|^2 \geq 
    \begin{cases}
     \frac{S_1-E}{S_1-E_0} - \frac{p}{\lambda_2}, &S^\top_1 <S_1,\\
     \frac{S_1-E}{S_1-E_0}, & \rm otherwise.
    \end{cases}
\end{equation}

\end{theorem}
\begin{proof}
    Using Theorem~\ref{theorem:overlap-unknown-alpha},
    \begin{equation}
    |\braket{\psi}{E_0}|^2 \geq 
    \begin{cases}
     \frac{S_1-E}{S_1-E_0} - \frac{p}{\lambda_2}\frac{S_1-S^\top_1}{S_1-E_0}, &S^\top_1 <S_1,\\
     \frac{S_1-E}{S_1-E_0} - \frac{p}{\lambda_P} \frac{S_1-S^\top_1}{S_1-E_0}, & \rm otherwise.
    \end{cases}
\end{equation}
For the case $S_1^\top < S_1$:
\begin{equation}
    \begin{split}
    |\braket{\psi}{E_0}|^2 & \geq \frac{S_1-E}{S_1-E_0} - \frac{p}{\lambda_2}\frac{S_1-S^\top_1}{S_1-E_0} \\
    &\geq \frac{S_1-E}{S_1-E_0} - \frac{p}{\lambda_2}\frac{S_1-E_0}{S_1-E_0} \\
    & \geq \frac{S_1-E}{S_1-E_0} - \frac{p}{\lambda_2}.
    \end{split}
\end{equation}
For the other case:
\begin{equation}
    \begin{split}
    |\braket{\psi}{E_0}|^2 & \geq \frac{S_1-E}{S_1-E_0} + \frac{p}{\lambda_P} \frac{S_1^\top-S_1}{S_1-E_0} \geq  \frac{S_1-E}{S_1-E_0}.
    \end{split}
\end{equation}
\end{proof}
If the relation between $S_1^\top$ and $S_1$ is unknown, one can always use the weaker bound $\frac{S_1-E}{S_1-E_0} - \frac{p}{\lambda_2}$.

\begin{table}[]
    \centering
    \begin{tabular}{cccc}
    \hline
     \#qubits & ADAPT-VMPE & lower bound & exact \\
     &  iterations & on overlap [\%] & overlap [\%] \\\hline
         \multirow{4}*{40 qubits}& 0 (HF) & 64.19 & 92.60 \\
         & 25 & 85.05 & 96.95  \\ 
         & 50 & 91.51 & 98.25  \\
         & 100 & 94.97 & 99.01  \\\hline
         \multirow{4}*{52 qubits}& 0 (HF) & 36.40 & 87.97  \\
         & 25 & 72.72 & 94.58 \\
         & 50 & 80.61 & 96.53  \\
         & 100 & 86.37 & 97.41  \\
         \hline
    \end{tabular}
    \caption{\textbf{Lower bound for 40 and 52 qubits and for HF and 25, 50 and 100 ADAPT-VMPE iterations.} We can see that both the lower bound on the overlap, and the exact overlap calculated with sparse-wavefunction statevector simulations of the ADAPT-VMPE circuit against DMRG outcome to dressed Hamiltonian are monotonically increasing with the number of iterations, however the latter are much larger.}
    \label{tab:overalps-comparison}
\end{table}

The lower bound computed with the Theorem above assumes the produced state is a particularly unfortunate superposition of ground and eigenstates, therefore it is expected that the real overlap is much higher that the estimated bound. We performed a separate analysis as describe below, and summarize our finding in Table~\ref{tab:overalps-comparison}. Our estimation of the exact overlap with the ground state was made by computing overlap between the DMRG solution and the sparse wavefunction simulation of circuit. The DMRG solution has been computed through the extraction of the determinants and related CI weights from the MPS wavefunction associated with the ground state of the actively rotated Hamiltonian. Following the same setup for the reference energies reported in Table~\ref{tab:tld-1433-ref-s0}, a non-spin adapted DMRG mode \texttt{SymmetryType.SZ} was used to solve the eigenvalue problem on the Hamiltonians for each of the two active spaces under analysis, i.e. 40 and 52 qubits. In particular, for the former we selected all the determinants with an absolute coefficient magnitude above $10^{-7}$, and $10^{-6}$ for the latter.

Since sparse wavefunction simulator cannot effectively generate a superposition of Slater determinants constructed with active rotations, we executed DMRG on the integrals dressed with active rotations found by ADAPT-VMPE. In order to obtain the dressed integrals, $\tilde{h}_{pq}$ and $\tilde{h}_{pqrs}$, the unitary transformation $U$ must be defined as 
\begin{equation}
    U= \prod_{k}e^{G_k \theta_k},
\end{equation}
where $G_k$ is the generator of the $k$-th orbital rotation between orbital $p$ and $q$ defined as $a_q^\dagger a_p - a_p^\dagger a_q$, as in Eq.~\eqref{eq:adapt-vmpe-output}. Then, the unitary can be applied to the original integrals, $h_{pq}$ and $h_{pqrs}$, obtaining the dressed integrals as
\begin{equation}
    \tilde{h}_{pq} = U_{ip} {h}_{ij} U_{jr}, \quad \tilde{h}_{pqrs} = U_{ip} U_{jq}{h}_{ijkl}U_{kr} U_{ls},
\end{equation}
where $h_{pq}$ is the original one-body tensor and $h_{pqrs}$ is the original two-body tensor (a more detailed explanation of the procedure can be found in Reference~\cite{wahtor2023}). In particular, the transformation $U$ is split between $\alpha$ and $\beta$ spin sectors. Following this rule, the transformation needs to be applied accordingly to the nature of the tensors, i.e. $\alpha$ transformation on $\alpha$ indexed tensors, $\beta$ transformation of $\beta$ indexed tensors.

The remaining part of the ADAPT-VMPE circuit was simulated with sparse wavefunction, truncating terms with coefficient magnitude $10^{-8}$ or less.

\end{document}